\documentclass[onecolumn, a4paper, accepted=2025-02-12]{quantumarticle}
\pdfoutput=1 

\newcommand{\be}{\begin{equation}}
\newcommand{\ee}{\end{equation}}

\newcommand{\calG}{\mathcal{G}(\omega)}

\newcommand{\Id}{\mathbb{I}}
\usepackage{float}

\usepackage[a4paper, total={7in, 9in}]{geometry}
\usepackage[utf8]{inputenc}
\usepackage{tikz}
\usepackage{caption}
\usepackage{graphicx}
\usepackage{subcaption}
\usepackage[normalem]{ulem}
\usepackage{quantikz}

\usepackage{amsthm}

\usepackage{longtable}

\usepackage{amsmath}
\usepackage{braket}
\usepackage{physics}

\usepackage{listings}
\usepackage[ruled,vlined]{algorithm2e}
\usepackage{amssymb}
\usepackage[ampersand]{easylist}
\usepackage{enumitem}
\usepackage{hyperref}
\usepackage[capitalise, nameinlink]{cleveref}
\usepackage[
  style=long,
  nolist, 
  nonumberlist,
  nopostdot,
  acronym,
  toc=true, 
]{glossaries}
\usepackage{bm}
\usepackage{booktabs}
\usepackage{multirow}

\usepackage{array}
\usepackage[british]{babel}
\usepackage{csquotes}
\usepackage{adjustbox}
\usepackage{comment}

\newtheorem{theorem}{Theorem}

\newtheorem{lemma}[theorem]{Lemma}

\usepackage[numbers]{natbib}

\usepackage{authblk}

\title{Approximating dynamical correlation functions with constant depth quantum circuits}
\date{June 6, 2024}

\author[1,2]{Reinis Irmejs}
\email{reinis.irmejs@mpq.mpg.de}
\author[1]{Raul A. Santos}
\email{raul@phasecraft.io}
\affil[1]{Phasecraft Ltd.}
\affil[2]{Max-Planck-Institut für Quantenoptik, Hans-Kopfermann-Straße 1, D-85748 Garching, Germany}

\begin{document}

\maketitle
\begin{abstract}

One of the most important quantities characterizing the microscopic properties of quantum systems are dynamical correlation functions. These correlations are obtained by time-evolving a perturbation of an eigenstate of the system, typically the ground state. In this work, we study approximations of these correlation functions that do not require time dynamics. We show that having access to a circuit that prepares an eigenstate of the Hamiltonian, it is possible to approximate the dynamical correlation functions up to exponential accuracy in the complex frequency domain $\omega=\Re(\omega)+i\Im(\omega)$, on a strip above the real line $\Im(\omega)=0$. 
We achieve this by exploiting the continued fraction representation of the dynamical correlation functions as functions of frequency $\omega$, where the level $k$ approximant can be obtained by measuring a weight $O(k)$ operator on the eigenstate of interest. In the complex $\omega$ plane, we show how this approach allows to determine approximations to correlation functions with accuracy that increases exponentially with $k$.

We analyse two algorithms to generate the continuous fraction representation in scalar or matrix form, starting from either one or many initial operators. We prove that these algorithms generate an exponentially accurate approximation of the dynamical correlation functions on a region sufficiently far away from the real frequency axis. We present numerical evidence of these theoretical results through simulations of small lattice systems. We comment on the stability of these algorithms with respect to sampling noise in the context of quantum simulation using quantum computers.

\end{abstract}

\tableofcontents

\section{Introduction}

A natural application of quantum computers is in the simulation of quantum systems. The resources needed to encode the Hilbert space, which scales exponentially with the system size, prohibit a direct classical simulation of these quantum systems. Of particular interest is the access to the time correlation functions: 
\be\label{eq:time_corr}
C(t)={\rm Tr}(\rho A(t) B) \quad \mbox{for $t>0$},
\ee
where $\rho$ is a state of the system and $A,B$ are operators describing some physical process. These correlation functions relate the response of a quantum system to external changes through the Kubo formula \cite{kubo1957statistical}, and are determined by its microscopic properties.

In physics, the dynamical correlation functions of this type are called Many-Body Green's functions (MBGF) \footnote{This name comes from the fact that in non-interacting systems, the time-ordered correlator of bosons (fermions) satisfies the usual differential equation of a Green's function. Specifically, the time-ordered correlation is defined as $-i\mathcal{T}\langle\psi_{\beta}(t)\psi_{\gamma}^{\dagger}\rangle:=\Theta(t)\langle\psi_{\beta}(t)\psi_{\gamma}^{\dagger}\rangle\pm\Theta(-t)\langle\psi^{\dagger}_{\gamma}\psi_{\beta}(t)\rangle$, where $\Theta(t)$ is the Heaviside step function, and the  $+$ ($-$) sign corresponds to $\psi$ being a bosonic (fermionic). If the time evolution of $\psi$ is generated by a quadratic Hamiltonian $H=\sum_{\alpha\beta} h_{\alpha\beta}\psi^\dagger_\alpha\psi_\beta$, then $-i\mathcal{T}\langle\psi_{\beta}(t)\psi_{\gamma}^{\dagger}\rangle$ satisfies
\be
\sum_\beta\left(i\frac{d}{dt}\delta_{\alpha\beta}-h_{\alpha\beta}\right)(-i\mathcal{T}\langle\psi_{\beta}(x,t)\psi_{\gamma}^{\dagger}(0)\rangle)=\delta(t)\delta_{\alpha\gamma}
\ee
}
 Knowing the MBGFs for $A$ and $B$ single boson (fermion) operators allows to obtain the ground state energy and momentum distribution of the particles \cite{abrikosov1997locally, fetter2012quantum}. The MBGFs also contain information about the excitation spectrum of the system, which can be seen using the Lehmann representation \cite{Kallen1952, lehmann1954eigenschaften} (see also discussion in \cref{sec:structure}). In classical simulation methods, MBGFs are a central object of investigation in embedded approaches to correlated systems. Dynamical mean field theory (DMFT) \cite{georges1996dynamical, kotliar2006electronic, aoki2014nonequilibrium}, one of the more advanced classical methods to study the behaviour of correlated materials, heavily hinges on the possibility of computing the time-dependent MBGFs. 

 Several algorithms have been proposed to simulate the MBGFs on quantum computers. These algorithms can largely be divided into 2 groups - those that perform direct time evolution and those based on subspace expansion methods. The approaches belonging to the first group involve appending an ancilla qubit to the system of interest and using the Hadamard test to access the real and imaginary parts of the MBGF obtained by computing the time-evolved state of the system subject to the state of the ancilla \cite{Bauer2016, Kreula2016, Kosugi2020, Ciavarella2020}. Thus, such a circuit requires both controlled time evolution and state preparation. To lower this high circuit cost and make these methods more suitable for near-term quantum devices, one can approximate the time evolution with variational circuits to reduce the depth as well as reduce the number of operators that need to have control \cite{Endo2020, Chen2021, Cruz2022, Sakurai2022, Libbi2022}. While being able to lower the necessary circuit depth, the variational approaches come at the price of dramatically increasing the number of necessary circuit runs to perform the optimization. Furthermore, the variational approaches lack the theoretical guarantees of faithfully approximating the dynamics. While at a lower cost than the direct approaches, we note that such algorithms still have a high cost of controlled operations on top of the already expensive ground state preparation, making them unlikely candidates for calculating the MBGFs on near-term devices.

The second strategy relies on calculating MBGFs directly in the frequency domain using Lehmann representation \cite{lehmann1954eigenschaften}. In this approach, one must calculate the overlaps of the operator of interest acting on the state $\rho$ of the system with each energy eigenstate (see \cref{sec:structure}). In practice, this approach is inefficient since one must calculate exponentially many overlaps. To make these methods feasible in the near term, one can use several approximations that allow one to estimate a small number of overlaps at the cost of limiting the accuracy of the MBGFs. In \cite{Jamet2022a}, the use of Krylov states instead of the excitations of the ground state was considered as it allows to more systematically truncate the number of the states needed. In this algorithm, the Krylov states are prepared with variational quantum algorithms (VQA), promising a faster MBGF convergence compared to excited states. In \cite{CQ_Lanczos}, authors use the Lanczos scheme to perform the subspace expansion. In their approach, they need to calculate the $\bra{\text{GS}}H^n\ket{\text{GS}}$ at increasing powers of $n$. This leads to the proliferation of, in general, exponentially many geometrically non-local observables as $n$ increases, which need to be evaluated. Small system investigations in \cite{CQ_Lanczos, Claudino_2021} have indicated that the proliferation of observables could saturate earlier and exhibit a sub-exponential scaling with $n$ for some practically relevant problems. The upside of the above approaches is that they can be implemented with much shallower circuits, as one only needs to perform measurements on the ground state in the respective basis. This leads to these approaches being more suitable for simulations in near-term quantum devices, as one only needs to perform very few operations on the hard-to-prepare ground-state approximations. These recursion methods have a long history in physics, for excellent resources on these approaches please see \cite{viswanath1994recursion,nandy2024quantum}.

In this work, we introduce an algorithm that belongs to the latter strategy. Here, we devise a method to classically obtain measurements that need to be performed on the state of interest $\rho$ to obtain an approximation of the desired accuracy of the MBGF. Furthermore, our algorithm guarantees the accuracy obtained with respect to the subspace expansion level chosen. The main result of this work, in its simplest formulation, is given by \cref{theo:approx_corr} below
and the associated extension presented in \cref{theo:approx_corr_mat} together with the algorithm of \cref{sec:algorithm}.
\begin{theorem}[Approximation of a correlation function]
\label{theo:approx_corr}
Given a quantum state $\rho$ and a Hamiltonian $H$ that commutes with $\rho$, the difference between the correlator $\mathcal{G}(\omega)=\int_0^\infty dt e^{i\omega t} {\rm Tr}(\rho A(t)A^\dagger)$ defined in the upper complex plane ${\rm Im}(\omega)>0$ and its $n$-th level continued fraction $\mathcal{G}_n(\omega)$ is bounded by 
\be\label{eq:approximation}
|\mathcal{G}({\omega})-\mathcal{G}_{n}({\omega})|\leq\frac{r}{1-2r}|\Gamma^{(n)}|\left|r(1-r)\frac{\Gamma^{(n-1)}}{\Gamma^{(n)}}\right|^{n+\frac{1}{2}},
\ee
for any $\omega$ such that ${\rm Im}(\omega)\geq\left[\frac{|\Gamma^{(n)}|}{|\Gamma^{(n-1)}|r(1-r)}\right]^{1/2}:= \Lambda$. Here $r$ is a free parameter $r\in(0,\frac{1}{2})$, $\Gamma^{(n)}={\rm Tr}(\rho A^{(n)}A^{(n)\dagger})$ and 
\be
A^{(n)}=[H,A^{(n-1)}]-\sum_{j=0}^{n-1}\frac{{\rm Tr}(\rho[H,A^{(n-1)}]A^{(j)\dagger})}{{\rm Tr}(\rho A^{(j)}A^{(j)\dagger})}A^{(j)}\quad \mbox{with}\quad A^{(0)}:= A.
\ee
The $n$-th level continued fraction $\mathcal{G}_n(\omega)$ is defined as the composition of M\"{o}bius maps
\be\label{eq:Mobius}
i\frac{\mathcal{G}_n(\omega)}{{\rm Tr}(\rho AA^\dagger)}= s_0^\omega \circ s_1^\omega \dots \circ s_{n-1}^\omega\circ s_n^\omega(0)\quad \mbox{with} \quad s^\omega_n(z)=\frac{-(\Gamma^{(n)})^2}{z+\Gamma^{(n)}\omega + \Delta^{(n)}},
\ee
and $\Delta^{(n)}={\rm Tr}(\rho [H,A^{(n)}]A^{(n)\dagger})$. This, by virtue of being a M\"{o}bius transformation, retains the Nevalinna structure \cite{nevanlinna1922asymptotische, Luger2014} of $\mathcal{G}(\omega)$. 
\end{theorem}
See \cref{app:Mobius_Nevalina} for more details on M\"{o}bius maps and how they preserve the Nevalinna structure of a function.

\begin{figure}
    \centering
    \includegraphics[width=\linewidth]{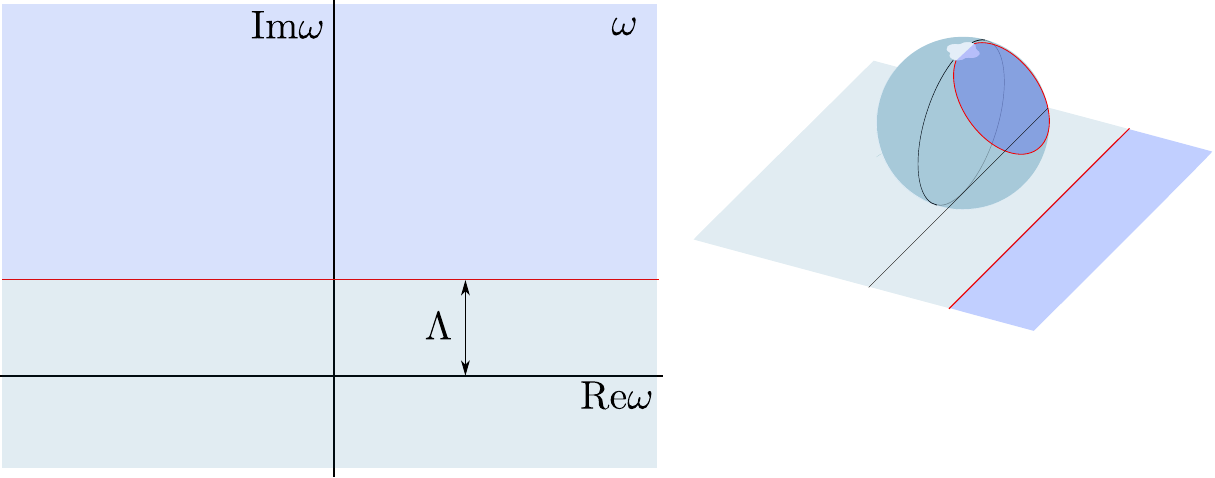}
    \caption{Left. The continued fraction approximation $\mathcal{G}_n(\omega)$ converges to $\mathcal{G}(\omega)$ exponentially fast with $n$ in a region above the real axis (light purple). Right. The M\"{o}bius maps are defined in the extended frequency plane with the point at infinity, that naturally maps into the Riemann sphere through a stereographic projection. The strip above the real axis in the plane becomes a cap in the sphere that contains the north pole.}
    \label{fig:Region_convergence}
\end{figure}

Notably, in the region ${\rm Im} (\omega)>\Lambda$ (see (\cref{fig:Region_convergence})), the accuracy of the approximation \cref{eq:approximation} increases exponentially with each level $n$.
Combining this approach with measurement techniques like classical shadows \cite{Huang_2020, Huang2021b, Zhao_2021} allows to construct $\mathcal{G}(\omega)$ in that region up to error $\epsilon$ in \cref{eq:approximation} using only a number of measurements that is polynomial in $1/\epsilon$.

We introduce two versions of this algorithm based on an expansion from a single and several operators. In particular, the single-operator algorithm is simpler but only allows calculating the time correlators of the form $\Tr(\rho A(t)A^\dagger)$. The many-operator approach is more involved but allows the calculation of $\Tr(\rho A_i(t)A_j^\dagger)$ where $A_i$ is an operator from a predefined (but arbitrary) set.

The bounds that we develop in this work are only provable when the observables of interest are evaluated with arbitrary precision. However, on real quantum hardware, one encounters various noise sources. Furthermore, even with no hardware noise, the accuracy is still limited by the statistical uncertainty arising from the finite number of circuit executions. In \cref{sec:Results} and \cref{app:results_app}, we investigate the sensitivity of MBGF numerically, as well as give theoretical bounds for the parameter $\eta$ for which the MBGF can be reliably evaluated for.

In both instances of the algorithm, the cost comes from performing the necessary measurements to evaluate the $\Gamma^{(n)}$ and $\Delta^{(n)}$ coefficients. The cost of measuring these coefficients depends on the support of their terms, which in both cases depends on the continued fraction level $n$. This allows to bound their measurements using the classical shadow approaches \cite{Huang_2020, Huang2021b, Zhao_2021}. In \cref{app:measurement}, we show that in the many-operator approach, the fermionic shadows \cite{Zhao_2021} can be used to evaluate the necessary measurements and bound the scaling of the algorithm. 

In \cref{sec:Algo}, we discuss the general structure of MBGFs and present their continued fraction representation in frequency space. An upper bound on the error induced by truncating that approximation is proved in \cref{sec:CFM}, and the full algorithm to approximate MBGFs is presented in \cref{sec:algorithm}. A simpler proof for the approximation of the single operator MBGF is discussed in \cref{sec:single_op_proof}.
Numerical simulations of the algorithm using a 1D Hubbard system are presented in \cref{sec:Results}. In \cref{sec:Discussion}, we discuss the algorithm's performance in practice, including in the presence of a noise and sampling error, followed by our conclusions in \cref{sec:Conclusions}.
\section{Continued fraction representation of Many-Body Green's functions}\label{sec:Algo}

\subsection{The structure of Many-Body Green's functions}\label{sec:structure}
The many-body Green's functions (MBGF) are defined by the correlator
\begin{align}
G_{AB}(t) :={\rm Tr}(\rho e^{iHt}Ae^{-iHt}B),
\end{align}
where $\rho$ is some density matrix, usually taken as the Gibbs state 
$\rho=e^{-\beta H }/{\rm Tr}e^{-\beta H}$. Using the Lehmann representation \cite{lehmann1954eigenschaften}, this correlator becomes in the frequency domain
\begin{align}\label{eq:GF_definition}
\mathcal{G}_{AB}(\omega):=\int_0^\infty G_{AB}(t)e^{i\omega t}dt=-i\sum_{m,n}\frac{\langle E_m|A|E_n\rangle \langle E_n|B\rho|E_m\rangle}{\omega+E_m-E_n},\quad \Im(\omega)>0.
\end{align}
 Here the states $|E_m\rangle $ are the eigenstates of the Hamiltonian $H$ with energy $E_m$. Although \cref{eq:GF_definition} is only defined for  $\Im(\omega)>0$, we can consider its analytic continuation to the whole complex $\omega$-plane, as this function is just the sum of simple poles at the resonance frequencies
$\omega_{mn}=E_n-E_m$.
 

 This representation of $\mathcal{G}_{AB}(\omega)$ reveals the usefulness of this quantity to study the system $\rho$ and the evolution generated by $H$, having access to operators $A$ and $B$ that act as probes. $\mathcal{G}_{AB}(\omega)$ will have a peak at $\omega = \omega_{mn}$, provided that the matrix elements of the operators in the numerator are non-zero.

 An important special case occurs if the state $\rho$ commutes with the Hamiltonian. In this scenario $\mathcal{G}_{AA^\dagger}(\omega)$ becomes
 \begin{align}\label{eq:Density_of_states}
     \mathcal{G}_{AA^\dagger}(\omega)=-i\sum_{m,n}\rho(E_m)\frac{|\langle E_m|A|E_n\rangle|^2}{\omega+E_m-E_n},
 \end{align}
 where, with a slight abuse of notation, we denote the eigenvalue of $\rho|E_m\rangle =\rho(E_m)|E_m\rangle$. Using the Sokhotski–Plemelj theorem on the real line, we find the relation
 \begin{align}
 \lim_{\eta \rightarrow 0^+} \Re(\mathcal{G}_{AA^\dagger})(\omega + i\eta)=\pi \sum_{m,n}\rho(E_m)|\langle E_m|A|E_n\rangle|^2\delta(\omega+E_m-E_n),\quad\mbox{for $w\in \mathbb{R}$}.
 \end{align}
If $\rho$ is a projector to the ground state of $H$, this corresponds to the spectral density, weighted by the matrix elements of $A$.

Motivated by the currently high costs of doing long-time dynamics on classical or quantum computers, we would like to develop an algorithm that can approximate the MBGF with finite depth circuits and understand its limitations. Finite time evolution up to time $T$ can approximate $\mathcal{G}_{AB}(\omega)$  on a strip in the complex plane $\Im(\omega)\geq\eta$ with an error
\begin{align}\label{eq:error}
\epsilon :=\left|\int_0^\infty e^{i\omega t}G_{AB}(t)dt-\int_0^T e^{i\omega t}G_{AB}(t)dt \right|\leq \frac{e^{-\eta T}|{\rm Tr}(\rho AB)|}{\eta}.
\end{align}
Clearly, allowing some nonzero imaginary part of the frequency $\omega$, transforms the sharp delta peaks in \cref{eq:Density_of_states} into Lorentzians with spectral resolution $\eta$, as
\begin{align}
\Re(\mathcal{G}_{AA^\dagger}(\omega + i\eta))=\sum_{m,n}\rho(E_m)\frac{|\langle E_m|A|E_n\rangle|^2\eta}{(\omega+E_m-E_n)^2+\eta^2},\quad\mbox{for $w\in \mathbb{R}$}.
\end{align}
This observation allows us to state the following lemma.
\begin{lemma}
The evolution time needed to approximate $\mathcal{G_{AA^\dagger}}(\omega)$ with error $\varepsilon$ and resolution $\eta$ of the spectral lines is bounded by
\begin{align}\label{eq:bound_time}
T\geq \frac{1}{\eta}\log(\frac{|{\rm Tr}(\rho AA^\dagger)|}{\varepsilon\eta}).
\end{align}
\end{lemma}
As the time evolved operator $A(t)$ possesses a series expansion given by $
    A(t)=\sum_{n=0}^\infty \frac{it}{n!}{\mathcal{L}}_H^n(A), $
where ${\mathcal{L}}^n_H(A)=[H,{\mathcal{L}}_H^{n-1}(A)]$, the bound \cref{eq:bound_time} suggests that it should be possible to approximate the MBGFs by truncating the evolution of $A$ in terms of finite number $N$ of nested commutators, with an error that decreases exponentially as the resolution decreases. Moreover, we look for an approximation that retains the Nevalinna structure \cite{nevanlinna1922asymptotische} of the MBGF, i.e. the approximation has a positive real part in the upper complex plane of $\omega$ if $\mathcal{G}_{AB}(\omega)$ has a positive real part in that region. The main result of this work is to show that such an approximation exists. To achieve this, we first generalise Mori's \cite{Mori1965} construction of a continued fraction representation of $G_{AA^\dagger}(\omega)$ to the case of many operators in the set $S_0=\{A_j\}_{j=1}^{N_{\rm op}}$ and then show that a truncation of the infinite continued fraction provides such approximation. We present these results in the following subsection.

\subsection{Continued fraction form of $\mathcal{G}_{ij}(\omega)=\int_0^\infty e^{i\omega t}{\rm Tr}(\rho A_i(t)A_j^\dagger)$}
\label{sec:CFM}

In order to generalise Mori's  \cite{Mori1965} results on the single operator MBGF $\mathcal{G}_{AA^\dagger}(\omega)$, to the case of $\mathcal{G}_{A_iA^\dagger_j}(\omega)$, where $A_{i},A_{j}$ belong to a set of operators $S_0:=\{A_j\}_{j=1}^{N_{\rm op}}$, we introduce the sesquilinear form
\begin{equation}\label{eq:innerGNS}
(X,Y):={\rm Tr}(\rho XY^{\dagger}),
\end{equation}
which through the Gelfand–Naimark–Segal (GNS) construction \cite{GelNeu43,Segal,doran1994c} defines a proper inner product. See \cref{app:GNS} for a discussion about this construction. Using this inner product and the associated induced norm, we can orthogonalize the set $S_0$, such that $(\hat{A}_j,\hat{A}_k)={\rm Tr}(\rho \hat{A}_j \hat{A}_k^{\dagger})=\delta_{jk}$ where
\begin{equation}\label{eq:norm_0}
    \hat{A}_j^{}:=\sum_m\left(\frac{1}{\sqrt{D}}U\right)_{jm}A_m\quad \mbox{with}\quad (A_j,A_k)=\sum_m U^{\dagger}_{jm}D_mU_{mk}.
\end{equation}
Here $U$ and $D$ are unitary and diagonal $N_{\rm op}\times N_{\rm op}$ matrices respectively. We define the projection onto the operator $\hat{A}_{j}$ as $P_{\hat{A}_{j}}(X):={{\rm Tr}(\rho X\hat{A}_{j}^{\dagger})}\hat{A}_{j}$,
which satisfies $P_{A_{j}}P_{A_{k}}=\delta_{jk}P_{A_{j}}$.
This projector is
Hermitian with respect to the inner product 
as 
\begin{align}
(P_{\hat{A}_{j}}(X),Y)&={\rm Tr}\left({{\rm Tr}(\rho X\hat{A}_{j}^{\dagger})}\rho \hat{A}_{j}Y^{\dagger}\right)
={\rm Tr}\left(\rho X\hat{A}_{j}^{\dagger}{{\rm Tr}(\rho \hat{A}_{j}Y^{\dagger})}\right)\\&={\rm Tr}(\rho X(P_{\hat{A}_{j}}(Y))^{\dagger})=(X,P_{\hat{A}_{j}}(Y)).
\end{align}
The evolution equation for the component of $A_p(t)$ orthogonal to $S_0$ is
\begin{equation}
\frac{d}{dt}Q_{0}(A_{p}(t))=iQ_{0}\mathcal{L}_{H}(Q_{0}(A_{p}(t)))+
i\sum_{j}{\rm Tr}(\rho A_{p}(t)\hat{A}_{j}^{\dagger})Q_{0}\mathcal{L}_{H}(\hat{A}_{j}),
\end{equation}
 which is simply obtained by projecting the original evolution equation
$\frac{dA_{p}(t)}{dt}=i\mathcal{L}_{H}A_{p}(t)$ to the space orthogonal to $S_0$, using the projector $Q_{0}:=1-\sum_{j}P_{\hat{A}_{j}}$. This equation can be formally solved as 
\begin{equation}
Q_{0}A_{p}(t)=i\sum_{j}\int_{0}^{t}d\tau{{\rm Tr}(\rho A_{p}(\tau)\hat{A}_{j}^{\dagger})}e^{i(t-\tau)Q_{0}\mathcal{L}_{H}}Q_{0}\mathcal{L}_{H}(\hat{A}_{j}).
\end{equation}
Now, the time evolution of operators in the orthogonal
complement of $S_0$, is determined by
the new operators $A_{j}^{(1)}(t):= e^{itQ_{0}\mathcal{L}_{H}}Q_{0}(\mathcal{L}_{H}(\hat{A}_{j}))$,
which are completely confined to the orthogonal complement of $S_{0}$.
We can repeat the previous procedure, now starting with the set of operators $S_{1}=\{\hat{A}^{(1)}\}_{j=1}^{N_{\rm op}}$, found by orthogonalizing the set $\{{A}^{(1)}(0)\}_{j=1}^{N_{\rm op}}$ and so on.
Repeating this procedure $k$-times splits the Hilbert space in the
orthogonal vector spaces spanned by the sets of operators $(S_{0},S_{1},S_{2},\dots S_{k-1})$
and their complement, where $S_{j}=\{\hat{A}_{p}^{(j)}\}_{p=1}^{N_{\rm op}}$ and
\begin{equation}
A_{p}^{(j)}:=\mathcal{L}_{j}\hat{A}_{p}^{(j-1)}=Q_{j-1}Q_{j-2}\dots Q_{1}Q_{0}\mathcal{L}_{H}\hat{A}_{p}^{(j-1)},\quad Q_{j}:=1-\sum_{k}P_{\hat{A}_{k}^{(j)}},\label{eq:c1_op}
\end{equation}
 (with $A_{p}^{(0)}:= A_{p}$). These vector spaces define a foliation
of the Hilbert space. The evolution in each subspace is determined
by the projected Liouvilian operator 
\begin{equation}
\mathcal{L}_{j}:=Q_{j-1}Q_{j-2}\dots Q_{1}Q_{0}\mathcal{L}_{H},
\end{equation}
 meaning that $A_{p}^{(j)}(t)= e^{it\mathcal{L}_{j}}A_{p}^{(j)}$. The normalised operators $\hat{A}^{(j)}_m$ are obtained in a similar way to \cref{eq:norm_0}, i.e.
 \begin{equation}\label{eq:norm_j}
    \hat{A}_p^{(j)}:=\sum_m\left(\frac{1}{\sqrt{D^{(j)}}}U^{(j)}\right)_{pm}A_m^{(j)}\rightarrow A_m^{(j)}=\sum_r U^{\dagger(j)}_{mr}\sqrt{D^{(j)}_r}\hat{A}_r^{(j)},
\end{equation}
with $(A_p^{(j)},A_l^{(j)})=\sum_m U^{(j)\dagger}_{pm}D^{(j)}_mU^{(j)}_{ml}$. Here $U^{(j)}$ is a unitary matrix, while $D^{(j)}$ is diagonal. It is convenient to define the matrix $M^{(j)}_{ab}=\sqrt{D^{(j)}}_aU_{ab}^{(j)}$, such that $(A_p^{(j)},A_l^{(j)})=\sum_m M^{\dagger(j)}_{pm}M_{ml}^{(j)}$.

With all these definitions in place, we can state:
\begin{theorem}[Continued fraction representation of $\bm{\mathcal{G}}(\omega)$]\label{theo:CFM}
The $N_{\rm op}\times N_{\rm op}$ matrix  $\bm{\mathcal{G}}(\omega)$ defined by its components as
$(\bm{\mathcal{G}}(\omega))_{ij}:=\int_0^\infty e^{i\omega t}{\rm Tr}(\rho A_i(t)A_j^\dagger) dt$ has a continued fraction representation in the form of the infinite iterated map
$\bm{\mathcal{G}}(\omega)= \lim_{k\rightarrow\infty} \bm{\mathcal{G}}_k(\omega)$ where 
\begin{equation}\label{eq:CFMobius}
    \bm{\mathcal{G}}_k(\omega):=\bm{s}^\omega_0\circ \bm{s}_1^\omega \circ \dots \circ \bm{s}_{k-1}^\omega\circ \bm{s}_k^\omega(0)\quad \mbox{and} \quad \bm{s}^\omega_j(\bm{X} ):=
- \left[\omega+\bm{\Delta}^{(j)} +\bm{M}^{\dagger(j+1)}\bm{X}\bm{M}^{(j+1)}\right]^{-1}.
\end{equation}
Here $\bm{\Delta}^{(j)}$ is a $N_{\rm op}\times N_{\rm op}$ matrix defined coefficientwise by  $(\bm{\Delta}^{(j)})_{km}:={\rm Tr}(\rho\mathcal{L}_j\hat{A}_k^{(j)}\hat{A}_m^{(j)\dagger})$.
The matrix $\bm{M}$ is obtained by diagonalising
\begin{align}
(\bm{R}^{(j)})_{pl}:=(A_p^{(j)},A_l^{(j)})=\sum_m M^{\dagger(j)}_{pm}M_{ml}^{(j)}.
\end{align}

Moreover, introducing the matrix M\"{o}bius map $T_{\bm{S}}(\bm{X}):=(\bm{A}\bm{X}+\bm{B})(\bm{C}\bm{X}+\bm{D})^{-1}$, where $\bm{S}$ is a $2N_{\rm op}\times 2N_{\rm op}$ matrix given by 
\begin{align}
   \bm{S}= \left(\begin{array}{cc}
    \bm{A}& \bm{B}\\
    \bm{C} & \bm{D}
    \end{array}\right),
\end{align}
the map $\bm{s}_j^\omega$ of \cref{eq:CFMobius} can be expressed as $\bm{s}^\omega_j=T_{\bm{S}_j(\omega)\bm{R}_{j+1}}$ with $\bm{S}_j(\omega)=\left(\begin{array}{cc}
   0  & -1 \\
   1  & \omega + \bm{\Delta}^{(j)}
\end{array}\right)$
and $\bm{R}_{j+1}=\left(\begin{array}{cc}
   \bm{M}^{\dagger(j+1)}  & 0 \\
   0  & (\bm{M}^{(j+1)})^{-1}
\end{array}\right)$. Then the MBGFs $\bm{\mathcal{G}}(\omega)$ takes the form
\begin{equation}
    \bm{\mathcal{G}}(\omega)= \lim_{k\rightarrow\infty} T_{\left(\prod_{j=0}^{k-1}\bm{S}_j(\omega)\bm{R}_{j+1}\right)\bm{S}_k(\omega)}(0).
\end{equation}

\end{theorem}
\begin{proof}
The proof of this theorem is a lengthy but straightforward computation, given in \cref{app:cont_frac_many}
\end{proof}

The previous theorem shows that there exists an infinite continued fraction representation of $\bm{\mathcal{G}}(\omega)$, and that this can also be represented as an infinitely iterated matrix M\"{o}bius map, but it does not show under which conditions this infinitely iterated map converges. In order to prove the convergence of the above limit, it is useful to define the approximants $\bm{P}_n\bm{Q}_n^{-1}:= \bm{\mathcal{G}}_n^{(0)}$.
By virtue of the composition rule $T_{\bm{S}_1}\circ T_{\bm{S}_2}(\bm{X}) := T_{\bm{S}_1}(T_{\bm{S}_2}(\bm{X}))=T_{\bm{S}_1\bm{S}_2}(\bm{X}) $, these approximants satisfy 
\begin{eqnarray}
 \left(\begin{array}{cc}
\bm{W}_{n} & \bm{P}_{n}\bm{M}^{(n+1)}\\
\bm{Z}_{n} & \bm{Q}_{n}\bm{M}^{(n+1)}
\end{array}\right)&:=&\left(\prod_{j=0}^{n-1}\bm{S}_{j}(\omega)\bm{R}_{j+1}\right)\bm{S}_{n}(\omega)\\&=&\left(\prod_{j=0}^{n-2}\bm{S}_{j}(\omega)\bm{R}_{j+1}\right)\bm{S}_{n-1}(\omega)\bm{R}_{n}\bm{S}_{n}(\omega)\\&=&\left(\begin{array}{cc}
\bm{W}_{n-1} & \bm{P}_{n-1}\bm{M}^{(n)}\\
\bm{Z}_{n-1} & \bm{Q}_{n-1}\bm{M}^{(n)}
\end{array}\right)\bm{R}_{n}\bm{S}_{n}(\omega),
\end{eqnarray}
where $\bm{W}_n$ and $\bm{Z}_n$ are auxiliary matrices. From here, we find the recurrence relation
\begin{equation}\label{eq:req_rel}
    \left(\begin{array}{c}
\bm{P}_{n}\bm{M}^{(n+1)}\\
\bm{Q}_{n}\bm{M}^{(n+1)}
\end{array}\right)=\left(\begin{array}{c}
\bm{P}_{n-1}(\omega+\bm{\Delta}^{(n)})-\bm{P}_{n-2}\bm{M}^{\dagger(n)}\\
\bm{Q}_{n-1}(\omega+\bm{\Delta}^{(n)})-\bm{Q}_{n-2}\bm{M}^{\dagger(n)}
\end{array}\right),
\end{equation}
with initial conditions $\bm{P}_{0}=-(\bm{M}^{(1)})^{-1},\bm{P}_{-1}=0$ and $\bm{Q}_{0}=(\omega+\bm{\Delta}^{(0)})(\bm{M}^{(1)})^{-1}, \bm{Q}_{-1}=1$.
This three-term recursion {\it defines} a matrix orthogonal polynomial \cite{Aptekarev_1984}. Here $\bm{P}_n(\omega)$ and $\bm{Q}_n(\omega)$ are matrix orthogonal polynomials of degree $n$ and $n+1$ respectively, meaning that the components of $\bm{P}_n(\omega)$ ($\bm{Q}_n(\omega)$) are regular polynomials (in the variable $\omega$) of maximum degree $n$ $(n+1)$.
As $Q_m(\omega)$ is a matrix orthogonal polynomial, we can make use of the following lemma:
\begin{lemma}[Surprising decomposition, Theorem 3.1 in \cite{Duran_1996}]\label{lem:surp}
The decomposition 
\begin{equation}
\bm{Q}_{m}^{-1}(\omega)\bm{R}(\omega)=\sum_{k=1}^{p}\frac{1}{\omega-x_{km}}\frac{({\rm Adj}(\bm{Q}_m(x_{km}))^{(l_k-1)}}{({\rm det} \bm {Q}_m)(x_{km}))^{(l_k)}}\bm{R}(x_{km}),
\end{equation}
is always possible for any matrix polynomial R of degree ${\rm deg}(R)\leq m-1$. Here $x_{km}$ $(1\leq k\leq p)$ is a zero of $Q_m(\omega)$ of multiplicity $l_k$. ${\rm Adj}(\bm A)(x)$ is the adjugate matrix of ${\bm A}(x)$ (also known as the classical adjoint, defined through the relation ${\bm A}{\rm Adj}(\bm A)={\rm det}(\bm{A})\bm{I}$) and $(f(a))^{(m)}$ is the $m$-th derivative of the $f$ evaluated at $a$.
\end{lemma}
This is a surprising fact as the zeros of $\bm{Q}_m(\omega)$, (i.e. the values $x$ solutions of ${\rm det}\bm{Q}_m(x)=0$) could have multiplicity more than one. With these definitions, we can state the main result of this work:

\begin{theorem}[Approximation of correlation functions]
\label{theo:approx_corr_mat}
Given a quantum state $\rho$, a Hamiltonian $H$ that commutes with $\rho$, and a set of orthonormal operators $\{\hat{A}_i\}_{i=1}^{N_{\rm op}}$ with respect to the inner product \cref{eq:innerGNS}, the difference between the correlator $(\bm{\mathcal{G}}(\omega))_{ab}=\int_0^\infty dt e^{i\omega t} {\rm Tr}(\rho \hat{A}_a(t)\hat{A}_b^\dagger)$ defined in the upper complex plane ${\rm Im}(\omega)>0$ and its $n$-th level continued fraction $(\bm{\mathcal{G}}_n(\omega))_{ab}$ is bounded by 
\begin{align}\label{eq:approximation_matrix}
||\bm{\mathcal{G}}(\omega)-\bm{\mathcal{G}}_n(\omega)||\leq\left(\frac{N_{\rm op}A}{\Lambda}\right)^{n+2}\frac{1}{1-\frac{N_{\rm op}A}{\Lambda}},
\end{align}
for any $\omega=\Re(\omega)+i\Im(\omega)$ such that ${\rm Im}(\omega):= \Lambda\geq N_{\rm op}A$ with \newline$A=\max_{m,k}\left(||\bm{A}_{km}\bm{M}^{\dagger(m+1)}(\bm{M}^{(m+2)})^{-1}||\right)$. Here 
\begin{align}
    \bm{A}_{km}=\frac{({\rm Adj}(\bm{Q}_m(x_{km}))^{(l_k-1)}}{({\rm det} \bm {Q}_m)(x_{km}))^{(l_k)}}\bm{Q}_{m-1}(x_{km}).
\end{align}

\end{theorem}
\begin{proof} Using the definition of $\bm{\mathcal{G}}_n$ in terms of $\bm{P}_n$ and $\bm{Q}_n$ and the recurrence relation \cref{eq:req_rel}, it is straightforward to show that the consecutive difference $\bm{\mathcal{G}}_n(\omega)-\bm{\mathcal{G}}_{n-1}(\omega)$ satisfies
\begin{align}\nonumber\label{eq:dif}
    \bm{\mathcal{G}}_n(\omega)-\bm{\mathcal{G}}_{n-1}(\omega)&=\bm{P}_{n}(\omega)\bm{Q}_{n}^{-1}(\omega)-\bm{P}_{n-1}(\omega)\bm{Q}_{n-1}^{-1}(\omega)\\
    &=(\bm{P}_{n-1}(\omega)\bm{Q}_{n-1}^{-1}(\omega)-\bm{P}_{n-2}(\omega)\bm{Q}_{n-2}^{-1}(\omega))\bm{Q}_{n-2}\bm{M}^{\dagger(n)}(\bm{M}^{(n+1)})^{-1}\bm{Q}_{n}^{-1}(\omega).
\end{align}
Iterating this procedure leads to
\begin{align}
\bm{\mathcal{G}}_n(\omega)-\bm{\mathcal{G}}_{n-1}(\omega)=-(\bm{M}^{(1)})^{-1}\left[\prod_{m=0}^{n-1}\bm{Q}_{m}^{-1}(\omega)\bm{Q}_{m-1}(\omega)\bm{M}^{\dagger(m+1)}(\bm{M}^{(m+2)})^{-1}\right]\bm{Q}_{n}^{-1}(\omega).
\end{align}
Making use of the triangle inequality, \cref{lem:surp}, and the fact that the zeros of orthogonal polynomials are real, we can bound the difference \cref{eq:dif} as
\begin{equation}
\| \bm{Q}_{m}^{-1}(\omega)\bm{Q}_{m-1}(\omega)\bm{M}^{\dagger(m+1)}(\bm{M}^{(m+2)})^{-1}\|\leq \sum_{k=1}^{r}\left\|\frac{\bm{A}_{k,m}\bm{M}^{\dagger(m+1)}(\bm{M}^{(m+2)})^{-1}}{\omega-x_{k,m}}\right\|\leq\frac{N_{\rm op}A}{\Lambda},
\end{equation}
leading to $\|\bm{\mathcal{G}}_n(\omega)-\bm{\mathcal{G}}_{n-1}(\omega)\|\leq \left(\frac{N_{\rm op}A}{\Lambda}\right)^{n+1}$. The limit of infinite iterations becomes 
\begin{align}
\bm{\mathcal{G}}(\omega)=\lim_{k\rightarrow\infty}\bm{P}_{k}(\omega)\bm{Q}_{k}^{-1}(\omega)=\sum_{k=n+1}^{\infty}\left(\bm{P}_{k}(\omega)\bm{Q}_{k}^{-1}(\omega)-\bm{P}_{k-1}(\omega)\bm{Q}_{k-1}^{-1}(\omega)\right)+\bm{P}_{n}(\omega)\bm{Q}_{n}^{-1}(\omega),
\end{align}
using that $\bm{\mathcal{G}}_{n}(\omega)=\bm{P}_{n}(\omega)\bm{Q}_{n}^{-1}(\omega)$ we finally find
\begin{align}
\|\bm{\mathcal{G}}(\omega)-\bm{\mathcal{G}}_{n}(\omega)\|	\leq\sum_{k=n+1}^{\infty}\|\bm{P}_{k}(\omega)\bm{Q}_{k}^{-1}(\omega)-\bm{P}_{k-1}(\omega)\bm{Q}_{k-1}^{-1}(\omega)\|
	\leq\left(\frac{N_{\rm op}A}{\Lambda}\right)^{n+2}\frac{1}{1-\frac{N_{\rm op}A}{\Lambda}},
\end{align}
which finishes the proof.
\end{proof}

\subsection{Algorithm for approximating $\mathcal{G}_{ij}(\omega)$}
\label{sec:algorithm}
Now that we have all the ingredients in place, the algorithm to approximate the MBGFs is simply
\newline

\hrule
\noindent \textbf{Algorithm to approximate $\bm{\hat{\mathcal{G}}}(\omega)$ $(\bm{\mathcal{G}}(\omega))$ with a continued fraction of $n$ levels}
\hrule
\,\newline

\noindent {\bf Set} $\{A^{0}_j\}_{j=1}^{N_{\rm op}}=\{A_j\}_{j=1}^{N_{\rm op}}$, $\mathcal{L}_0(*)=[H,*]$ and $p=0$

\noindent\textbf{While} $p \leq n$ {\bf do:}
\begin{enumerate}
    \item Given the set of operators $\{A_j^{(p)}\}_{j=1}^{N_{\rm op}}$ and the superoperator $\mathcal{L}_p[*]$ calculate the matrix $\bm{R}$ with components
        $R_{ab}={\rm Tr}(\rho A_a^{(p)} A^{\dagger(p)}_b)$.
    \item Diagonalize $\bm{R}=\bm{U}^{\dagger(p)}\bm {D}^{(p)} \bm{U}^{(p)}=\bm{M}^{\dagger(p)}\bm{M}^{(p)}$ where $\bm{D}$ is diagonal and $\bm{U}^{(p)}$ is unitary.
    \item Define the normalized operators \begin{equation}
        \hat{A}_p^{(j)}=
        \begin{cases}
        \sum_m\frac{1}{\sqrt{D^{(j)}_p}}U^{(j)}_{pm}A_m^{(j)} & \text{if $D^{(j)}_p>0$}\\
        0 & \text{otherwise.}
        \end{cases}
    \end{equation}
    \item Obtain the matrix $\bm{\Delta}^{(p)}$ with coefficients $\Delta_{km}^{(p)}={\rm Tr}(\rho \mathcal{L}_p [\hat{A}_k^{(p)}]\hat{A}^{\dagger(p)}_m) $.
    \item Update 
    \begin{eqnarray}
        \{A^{(p+1)}_k\}_{k=1}^{N_{\rm op}} &=&\{\mathcal{L}_p \hat{A}^{(p)}_k\}_{k=1}^{N_{\rm op}},\\
        \mathcal{L}_{p+1}[*]&=&\left(\mathcal{L}_{p}[*]-\sum_{k=1}^{N_{\rm op}}{\rm Tr}(\rho\mathcal{L}_p[*]\hat{A}_k^{\dagger(p)})\hat{A}_k^{(p)}\right).
    \end{eqnarray}
    \item Go back to step 1 with $p\rightarrow p+1$.
\end{enumerate}
\noindent \textbf{Construct the iterated map}
\begin{equation}
\bm{\hat{\mathcal{G}}}_{n}=\bm{s}^\omega_0\circ \bm{\hat{s}}_1^\omega \circ \dots \circ \bm{\hat{s}}_{n-1}^\omega\circ \bm{\hat{s}}_{n}^\omega(0) \quad \mbox{with} \quad \bm{\hat{s}}^\omega_j(\bm{X} )=- \left[\omega+\bm{\Delta}^{(j)} +\bm{M}^{\dagger(j+1)}\bm{X}\bm{M}^{(j+1)}\right]^{-1},
\end{equation}
to generate the approximation of $n$ levels to $\bm{\hat\mathcal{{G}}}$ with matrix elements $\hat{\mathcal{G}}_{ab}(\omega)=\int_0^\infty{\rm Tr}(\rho \hat{A}_a(t)\hat{A}_b)$ obtained from the {\it normalized} initial set of operators.

Equivalently, the iterated map
\begin{equation}
\bm{\mathcal{G}}_n=\bm{s}^\omega_0\circ \bm{s}_1^\omega \circ \dots \circ \bm{s}_{n-1}^\omega\circ \bm{\hat{s}}_{n}^\omega(0)\quad \mbox{with} \quad \bm{s}^\omega_j(\bm{X} )=- \bm{M}^{\dagger(j)}\left[\omega+\bm{\Delta}^{(j)} +\bm{X}\right]^{-1}\bm{M}^{(j)},
\end{equation}
generates an approximation  of $n$ levels to the {\it unnormalized} matrix Green's function $\bm{\mathcal{G}}$ obtained from the initial set of operators $\{A_j\}_{j=1}^{N_{\rm op}}$.
\hrule
\,\newline

Using a single operator in the initial set $S_0=\{A\}$, it is possible to simplify the previous results, opening the possibility of using the extensive results about orthogonal polynomials to derive a convergence proof of the continued fraction approach. In the single operator case we define $\mathcal{G}(\omega):=\mathcal{G}_{AA^\dagger}(\omega)$, which is defined as the infinitely iterated M\"{o}bius transformation $\mathcal{G}(\omega)= \lim_{k\rightarrow\infty} \mathcal{G}_k(\omega)$ where
\begin{equation}\label{eq:cont_frac_mob}
    \mathcal{G}_k(\omega)=s^\omega_0\circ s_1^\omega \circ \dots \circ s_{k-1}^\omega\circ s_k^\omega(0)\quad \mbox{and} \quad s^\omega_j(z)=\frac{-(\Gamma^{(j)})^2}{z+\Delta^{(j)}+\Gamma^{(j)}\omega}.
\end{equation}
The proof of \cref{theo:approx_corr} then proceeds as follows:
\subsection{Proof of \cref{theo:approx_corr}: Continued fraction representation, single operator}\label{sec:single_op_proof}

\begin{proof}

We want to find the error in approximating $\mathcal{G}(\omega)= \lim_{k\rightarrow\infty} \mathcal{G}_k(\omega)$ with $\mathcal{G}_n(\omega)$ for finite $n$. 
Adapting the results from Merkes \cite{Merkes1966} we find that if $\forall n\geq 2$
\begin{equation}\label{eq:condition}
    \left|\frac{(\Gamma^{(n)})^2}{(\Gamma^{(n)}\omega+\Delta^{(n)})(\Gamma^{(n-1)}\omega+\Delta^{(n-1)})}\right|\leq r(1-r)\quad \mbox{with} \quad r\in(0,\frac{1}{2}),
\end{equation}
the {\it posteriory bound} is satisfied
\begin{equation}\label{eq:post_bound}
    |\mathcal{G}({\omega})-\mathcal{G}_n({\omega})|\leq \frac{r}{1-2r}|\mathcal{G}_n({\omega})-\mathcal{G}_{n-1}({\omega})|.
\end{equation}

To bound the difference $|\mathcal{G}_n({\omega})-\mathcal{G}_{n-1}({\omega})|$, we recall that the numerator $A_n(\omega)$ and denominator $B_n(\omega)$ of continued fractions $\mathcal{G}_n(\omega)=\frac{A_n(\omega)}{B_n(\omega)}$ defined by \cref{eq:cont_frac_mob} satisfy the recurrence relation
\cite{Koornwinder2013}
\begin{equation}\label{eq:rec_rel}
    Y_{n+1}(\omega) = (\Gamma^{(n+1)}\omega+\Delta^{(n+1)})Y_{n}(\omega)-(\Gamma^{(n+1)})^2Y_{n-1}(\omega) \quad \mbox{for}\quad n\geq 1,
\end{equation}
with $Y_n(\omega)=A_n(\omega), B_n(\omega)$ and initial conditions $A_{-1}(\omega)=0$, $B_{-1}(\omega)=1$ and $A_{0}(\omega)=-(\Gamma^{(0)})^2$, $B_0(\omega)=\Gamma^{(0)}\omega+\Delta^{(0)}$. Thus we have 
\begin{align}
    \mathcal{G}_n({\omega})-\mathcal{G}_{n-1}({\omega})&=\frac{1}{B_n(\omega)B_{n-1}(\omega)}\left|\begin{array}{cc}
    A_n(\omega) & A_{n-1}(\omega) \\
    B_n(\omega) & B_{n-1}(\omega) \end{array}\right|\\&=\frac{|\Gamma^{(n)}|^2}{B_n(\omega)B_{n-1}(\omega)}\left|\begin{array}{cc}
    A_{n-1}(\omega) & A_{n-2}(\omega) \\
    B_{n-1}(\omega) & B_{n-2}(\omega) \end{array}\right|,
\end{align}
where we have used the recursion relation \cref{eq:rec_rel} and the properties of the determinant. This leads to
\begin{equation}\label{eq:suc_bound}
    \left|\mathcal{G}_n(\omega)-\mathcal{G}_{n-1}(\omega)\right|=\frac{\prod_{j=0}^n|\Gamma^{(j)}|^2}{|B_n(\omega)B_{n-1}(\omega)|},
\end{equation}

While $\Gamma^{(n)}$ is always a real number (moreover $\Gamma^{(n)}={\rm Tr}(\rho A^{(n)}A^{(n)\dagger})\geq 0$), $\Delta^{(n)}$ is real only when the Liouvillian $\mathcal{L}_H$ operator is Hermitian. With this condition (satisfied when $[\rho, H]=0$), the denominator $B_n(\omega)$ is an orthogonal polynomial, meaning that there exists a measure function $\mu>0$  on $\mathbb{R}$ such that the polynomials $B_n(\omega)$ are orthogonal with respect to $\mu$, due to the three-term recursion relation \cref{eq:rec_rel} \cite{chihara2014introduction}. Let's recall that an orthogonal polynomial of degree $n$ has exactly $n$ real roots. We can use this to finally bound \cref{eq:suc_bound}. Noting that the largest power of $\omega$ in $B_n(\omega)$ has a factor $\prod_{j=0}^n\Gamma^{(j)}$, it follows that we can write $B_n(\omega)$ as a product of its roots and find
\be
B_n(\omega)=\left(\prod_{j=0}^n\Gamma^{(j)}\right)|(\omega-\omega_0)(\omega-\omega_1)\dots(\omega - \omega_n)|\geq \Lambda^{n+1}\prod_{j=0}^n\Gamma^{(j)}
\ee
for frequencies in the region ${\rm Im}(\omega)\geq\Lambda$, thus leading to
\begin{equation}\label{eq:suc_bound_final}
    \left|\mathcal{G}_n(\omega)-\mathcal{G}_{n-1}(\omega)\right|\leq\frac{\Gamma^{(n)}}{\Lambda^{2n+1}}.
\end{equation}
To use this result to bound the error on the approximation \cref{eq:post_bound} to the infinite continued fraction, we need to choose 
\be
\Lambda \geq \left[\frac{|\Gamma^{(n)}|}{|\Gamma^{(n-1)}|r(1-r)}\right]^{1/2},
\ee
such that \cref{eq:condition} is satisfied. This leads to 
\be
\left|\mathcal{G}({\omega})-\mathcal{G}_{n}({\omega})\right|\leq\frac{r}{1-2r}\Gamma^{(n)}\left|r(1-r)\frac{\Gamma^{(n-1)}}{\Gamma^{(n)}}\right|^{n+\frac{1}{2}},
\ee
the parameter $r\in(0,\frac{1}{2})$ could be optimised for a given instance. Note that $r\rightarrow 0$ improves the error of the approximation but reduces the region in the complex plane where the approximation is useful. On the other hand, $r\rightarrow \frac{1}{2}$ reduces $\Lambda$ but makes the approximation error blow up. This result proves \cref{theo:approx_corr}.
\end{proof}

These results confirm that an approximation of the MBGFs with the correct analytical structure is possible by relaxing the precision where its peaks can be resolved. If the resolution is too poor, the approximation, although still valid, could be of little use. In order to test the predictive power of this method, in the next section, we perform numerical simulations of the algorithm presented in \cref{sec:algorithm} for some examples.

\section{Results}\label{sec:Results}

\subsection{Perfect execution of the algorithm}\label{sec:perfect}

In this section, we numerically test the continued fraction algorithm presented in \cref{sec:algorithm} and verify the exponential improvement of the error with increasing continued fraction level established in \cref{sec:CFM}. We test the algorithm in the one-dimensional Hubbard model \cref{eq:FH_Ham} and compare the results against exact diagonalization.
The $N$-site Fermi-Hubbard model is defined as:
\begin{equation}\label{eq:FH_Ham}
    H = -t \sum_{i=0}^{N-2} \sum_{\sigma = \uparrow, \downarrow}(a_{i\sigma}^{ \dagger}a_{i+1 \sigma}+a_{i+1 \sigma}^{\dagger}a_{i \sigma}) + U\sum_{i=0}^{N-1}a_{i \uparrow}^ {\dagger}a_{i \uparrow} a_{i \downarrow}^{ \dagger}a_{i \downarrow},
\end{equation}
where $a_{i{\sigma}}$ $(a_{i{\sigma}}^\dagger)$ denotes the fermionic annihilation (creation) operator on site $i$ of a particle with spin $\sigma$. Here $t$ represents the strength of the hopping terms, while $U$ parameterises the on-site interaction. 

In the presence of a symmetry of the Hamiltonian, the set $S_0$ that defines the pool of operators from where $\mathcal{G}_{ij}(\omega)$ is constructed could transform non-trivially. For a discrete symmetry, realised through the unitary transformation $U_s$, it is always possible to append new operators to $S_0$ such that $S_0$ becomes a {\it normal} operator set with respect to the symmetry, i.e., $A_i^s:=U_sA_iU_s^\dagger\in S_0$ for any $A_i$. For a normal set $S_0$, the MBGFs satisfy the symmetry relation $\mathcal{G}^s_{ij}(\omega):=\int_{0}^{\infty}dt e^{i\omega t} \Tr(\rho A^s_i(t)A_j^{s\dagger})=\mathcal{G}_{ij}(\omega)$, as long as $\rho$ does not transform under the symmetry. We note that by this procedure, the number of measurements can be reduced, as we could {\it  assume} that the original set $S_0$ is a normal operator set, and use the relation $\mathcal{G}^s_{ij}=\mathcal{G}_{ij}$ to obtain the MBGF of the operators $A_i^s$ by using one of the operators $A_i$. On the other hand, we could construct the set $S_0$ such that it is normal, and perform all the measurements needed, increasing the measurement cost, but use the symmetry relation to mitigate errors in near term computations.

We study the real part of the MBGF defined as
\begin{equation}
     {\rm Re}(\mathcal{G}_{ij}(\omega)) = \int_{0}^{\infty}dt {\rm Re}(e^{i\omega t} \Tr(\rho A_i(t)A_j^\dagger)),
\end{equation}
where $A_i\in\{a_i\}_{i=0}^{N-1}$. For concreteness, we take the state $\rho$ as the projector onto the ground-state $\rho=|{\rm GS}\rangle\langle{\rm GS}|$ at half filling. As the real part of the MBGF satisfies $\Re(\mathcal{G}_{ij}(\omega)) = \Re(\mathcal{G}_{ji}(\omega))$, if $\comm{\rho}{H} = 0$ (see \cref{sec:results_app}), we concentrate in $\mathcal{G}_{ij}(\omega)$, for $i\leq j$. The particular Hamiltonian of \cref{eq:FH_Ham} is also invariant under reflection around the middle bond (for $N$ even) or the middle site (for $N$ odd), so we can further restrict by using the symmetry $\mathcal{G}_{ij}(\omega) = \mathcal{G}_{(N-i)(N-j)}(\omega)$. 

In \cref{fig:fig1_u4}, we consider the $N=4$ site Fermi-Hubbard model, with the set of initial operators $\{A_j\}_{j=0}^{N-1} = \{a_{j \uparrow}\}$ for $U/t = 4$. We plot the real part of $\mathcal{G}_{ij}(\omega)$ for the independent $i,j$ pairs for different continued fraction levels $n=0,1,2,3$, together with the exact diagonalization result. For each case, we also plot the error between the exact signal and the approximated continued fraction in logarithmic scale to highlight the exponential improvement of the signal as the level of the continued fraction is increased. 

\begin{figure}[ht!]
    \centering
    \includegraphics[width = \linewidth]{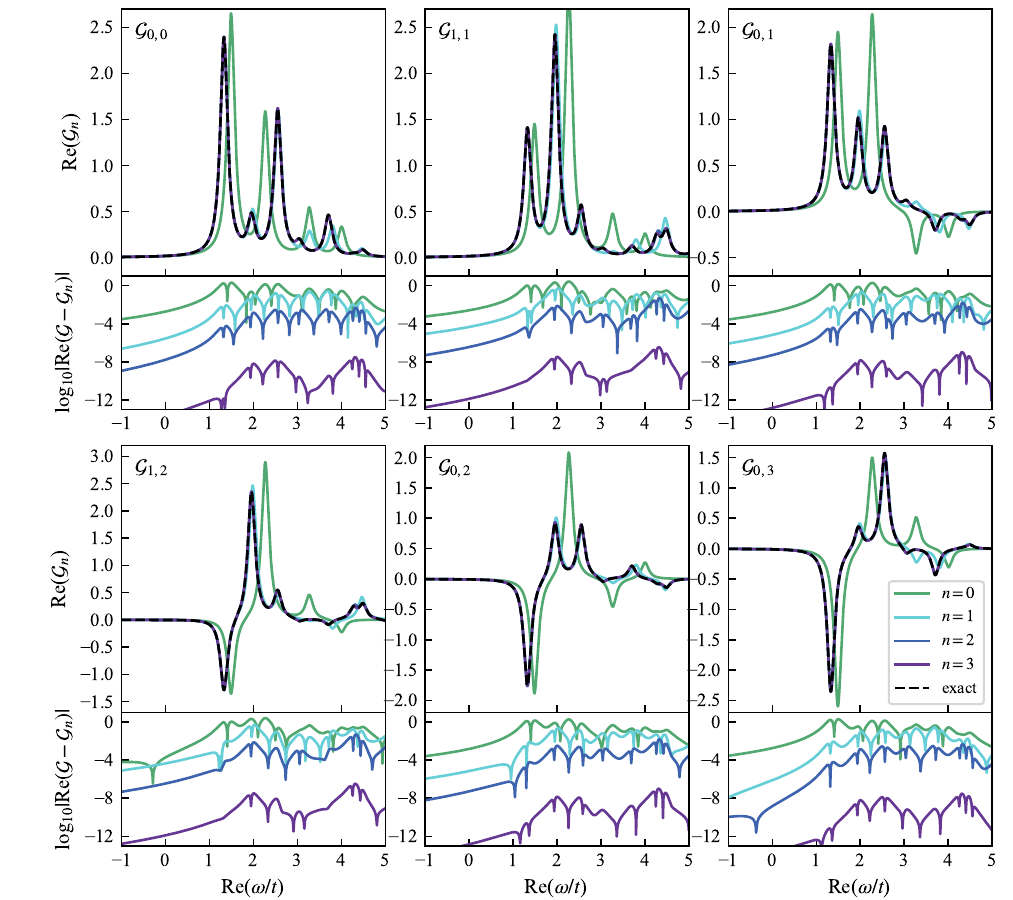}
    \caption{Approximation of the real part of the Fourier transformed MBGF obtained with the continued fraction algorithm at different truncation levels $n$ for the 4-site Fermi-Hubbard model with $U/t = 4$ with the initial operator choice $S_0= \{a_{j \uparrow}\}_{j=0}^{N-1}$. We plot the results for $\omega=\omega_0+i\eta$ with $\eta = 0.1$ and compare them with the MBGF function obtained using exact diagonalization. Different panels correspond to different choices of sites in $\mathcal{G}_{ij}$. Below each result for ${\rm Re}((\bm{\mathcal{G}}_n)_{ij})$, we show the error with respect to the exact result in logarithmic scale for the different truncation levels.}
    \label{fig:fig1_u4}
\end{figure}

Results for $U/t=2$ are presented in \cref{sec:results_app}, where the same conclusions apply.

\subsection{Execution on a real quantum computer}

While the algorithm presented in \cref{sec:algorithm} does not, in principle, require a quantum device, it is probably in systems beyond classical treatment where it may prove more useful. With this motivation, we test the algorithm in a small instance of the Fermi-Hubbard model, the Hubbard dimer on two sites, at half filling and $S^{\rm total}_z=0$, where we used a compact representation of the Hamiltonian 
\begin{equation}
    H_d = \frac{U}{2}(\Id_0 \Id_1 +Z_0 Z_1) - t(\Id_0 X_1 + X_0 \Id_1),
\end{equation}
in the basis $|{0_\uparrow},{0_\downarrow}\rangle,$ $|{0_\uparrow},{1_\downarrow}\rangle$,
$|{1_\uparrow},{0_\downarrow}\rangle$ and $|{1_\uparrow},{1_\downarrow}\rangle$, where $|i_\uparrow,j_\downarrow\rangle:=a^\dagger_{i\uparrow}a_{j\downarrow}|0\rangle$ and $|0\rangle$ is the trivial vacuum. With this encoding, the number operator in the first site $a_{0\uparrow}^\dagger a_{0\uparrow}=n_{0\uparrow}$ corresponds to $\frac{1}{2}(\Id_0+Z_0)\Id_1$. We use this single operator to build up the continued fraction approximation of $\mathcal{G}(\omega)=\int_0^\infty e^{i\omega t}{\rm Tr}(\rho n_{0\uparrow}(t)n_{0\uparrow}) dt$. Finding the ground state $\ket{\psi_0}$ and the associated energy $E_0$ of this model is straightforward due to its small size:
\begin{equation}
    \ket{\psi_0} = \frac{1}{\sqrt{2\mathcal{N}}}(\alpha (\ket{00}+\ket{11}) + \beta({\ket{01}+\ket{10}})), \quad E_0 = U/2 -\sqrt{U^2/4 + 4t^2},
\end{equation}
where $\alpha = 4$, $\beta = U/t + \sqrt{U^2/t^2 +16}$, $\mathcal{N} = \alpha^2 + \beta^2$. The state $\rho = \ket{\psi_0}\bra{\psi_0}$ corresponds to the ground state of the dimer, which is prepared using a circuit of the form presented in \cref{fig:dimer_circuit} and pre-optimized parameters. The experiments were conducted on the 'IBM Jakarta' quantum device. 

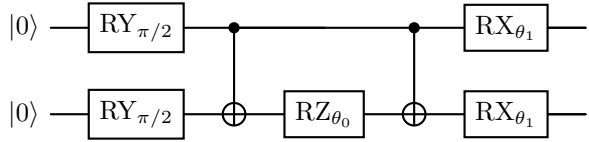
\begin{figure}[ht!]
    \centering
    \begin{quantikz}
        \lstick{$\ket{0}$} &  \gate{\text{RY}_{\pi/2}} & \ctrl{1} & \qw & \ctrl{1} & \gate{\text{RX}_{\theta_1}} & \qw\\
        \lstick{$\ket{0}$} & \gate{\text{RY}_{\pi/2}} & \targ{1} & \gate{\text{RZ}_{\theta_0}} & \targ{1} & \gate{\text{RX}_{\theta_1}} & \qw\\
    \end{quantikz}
    \caption{Circuit implementing the ground state for the Fermi-Hubbard dimer in the compact representation. The parameters $\theta_i$ depend on the relative coupling strength $U/t$. In particular, we can set $\theta_1 = \pi/4$, which leads to $\cos(\theta_0)/(1-\sin(\theta_0)) = \alpha/\beta$.}
    \label{fig:dimer_circuit}
\end{figure}
 
In \cref{fig:IBMQ}, we plot the continued fraction approximation at the level $n=4$. At this level, the algorithm is able to obtain the exact solution. The continued fraction algorithm requires to estimate the expectation of the Paulis $O_i = \{Z_0, X_0, X_0Z_1, Z_0Z_1, Y_0Y_1, X_1, Y_0, Y_0Z_1, X_0Y_1, Z_0Y_1 \}$. Further simplification of the measurement protocol can be obtained from the fact that the coefficients $\{\Gamma^{(j)}, \Delta^{(j)}\}$ must be real. Expanding the $A^{(j)}$ operator in the Pauli basis $O_i$, we find the linear equations
\begin{equation}\label{eq:linear_system}
    \Gamma^{(j)} = \sum_m g_{jm}{\rm Tr}(\rho O_m),\quad \Delta^{(j)} = \sum_m d_{jm}{\rm Tr}(\rho O_m).
\end{equation}
The reality conditions and the fact that the operators $O_i$ are hermitian, imply that \newline $\sum_m \Im{g_{jm}}{\rm Tr}(\rho O_m)=\sum_m \Im{d_{jm}}{\rm Tr}(\rho O_m)=0$, for each $j$. In the dimer example, this fixes the expectation of the operators $\{Y_0, Y_0Z_1, X_0Y_1, Z_0Y_1 \}$ to zero.
To measure each of the remaining $O_i$ at least $\mathcal{M}$ times, we need to only measure $ \{X_0X_1, X_0Z_1, Y_0Y_1, Z_0Z_1\}$ each $\mathcal{M}$ times. We perform these measurements directly by applying a suitable rotation for each of the four different basis we need to measure in. For the experimental results of \cref{fig:IBMQ}, we took $\mathcal{M}=32000$, which resulted in the total measurement cost of $4\mathcal{M} = 128000$.

\begin{figure}[ht!]
    \centering
    \includegraphics[width=\linewidth]{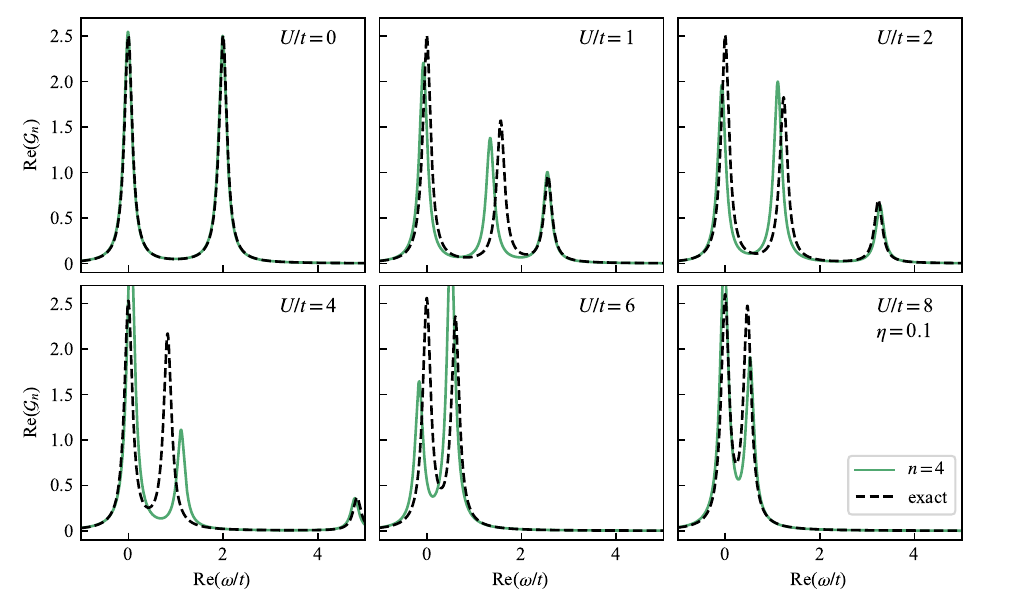}
    \caption{Approximation of the real part of the Fourier transformed MBGF obtained with the continued fraction algorithm for the 2-site Fermi-Hubbard model with the initial operator $A_0 = a_{0\uparrow}^\dagger a_{0\uparrow}$. We implement the exact ground state of the model using a pre-optimized VQE. The results are plotted for different interaction strengths $U/t = (0,1,2,4,6,8)$. Each necessary observable is measured at least 32000 times on the IBM Jakarta quantum device. We plot the results at the $n=4$ continued fraction level and for $\eta = 0.1$ and compare them with the exact MBGF. }
    \label{fig:IBMQ} 
\end{figure}

\subsection{Effect of state preparation and sampling noise}

So far, we have assumed that the desired state $\rho$ is available to us without an error. In practice, this may not be the case, particularly if $\rho$ is difficult to prepare, something that is generally QMA hard \cite{Kempe2006} for the ground state.

In this section, we will investigate the algorithm's accuracy by considering imperfect parametrization of the state $\rho$.
Let's consider a state $\tilde{\rho}:=(1-\epsilon)\rho+\epsilon \sigma$, which is a convex combination of a state $\rho$ that commutes with the Hamiltonian $H$ and some spurious state $\sigma$ that does not. The exact MBGF, being a linear functional of the state, becomes simply (from \cref{eq:GF_definition})
\begin{align}
\tilde{\mathcal{G}}_{AA^\dagger}(\omega):=-i\sum_{m,n}\frac{\langle E_m|A|E_n\rangle \langle E_n|A^\dagger\tilde{\rho}|E_m\rangle}{\omega+E_m-E_n}=(1-\epsilon)\mathcal{G}_{AA^\dagger}(\omega)-i\epsilon\sum_{m,n}\frac{\langle E_m|A|E_n\rangle \langle E_n|A^\dagger\sigma|E_m\rangle}{\omega+E_m-E_n},
\end{align}
where we see that if the error between the imperfect preparation and the exact one is $O(\epsilon)$, then $|\tilde{\mathcal{G}}_{AB}(\omega)-{\mathcal{G}}_{AB}(\omega)|=O(\epsilon)$. The continued fraction algorithm, though, relies on the condition $[\rho,H]=0$. Following the derivation of the algorithm (see \cref{app:single_op}), it is straightforward to see that the error produced by the state $\tilde{\rho}$ translates into the modified recursion for the continued fraction that defines $\mathcal{G}_{AA^\dagger}(\omega)$
\begin{align}\label{eq:tilde_G}
 i\tilde{\mathcal{G}}^{(j)}(\omega)=\frac{-(\Gamma^{(j)})^{2}}{\Gamma^{(j)}\omega+\Delta^{(j)}+i(1-\epsilon)\mathcal{G}^{(j+1)}(\omega)+i\epsilon \Sigma^{(j+1)}(\omega)}= i{\mathcal{G}}^{(j)}(\omega)\left(1+O(\epsilon)\right).
\end{align}

This relative error could compound through the recursive iterations, rendering the algorithm invalid. Here, we show that the algorithm based on a single operator is stable for a small enough deviation from a state that commutes with $H$ and sufficiently small shot noise. We also perform numerical simulations on the 4-site Fermi-Hubbard model to study the stability of the many-operator algorithm in practice. The main tool to show the algorithm's stability is a classic result in the theory of continued fractions, Theorem 3.1 in \cite{Jones1974}, that we revisit in our context.

Let us denote by $\tilde{\Gamma}^{(j)}$ and $\tilde{\Delta}^{(j)}$ the values of the estimated constants in a real experiment and the relative errors with respect to the exact $\Gamma^{(j)}$ and $\Delta^{(j)}$ by  $\gamma^{(j)}$ and $\delta^{(j)}$ such that
\begin{align}
\tilde{\Gamma}^{(j)}=\Gamma^{(j)}(1+\gamma^{(j)}), \quad \tilde{\Delta}^{(j)}=\Delta^{(j)}(1+\delta^{(j)}).
\end{align}
Here, $\delta^{(j)},\gamma^{(j)}$ parameterise the general errors that can come from imperfect estimation due to shot noise or algebraic errors due to machine error in the computation of the nested commutators. In the same way, we define the relative error $\varepsilon^{(j)}$ associated with computing $\mathcal{G}^{(j)}$ using all the $\tilde{\Gamma}^{(k)},\tilde{\Delta}^{(k)}$ for $j\leq k< n$, where $n$ is the truncation level (i.e $\mathcal{G}^{n+1}(\omega)=0$ as
\begin{align}\label{eq:G_err_gen}
i\tilde{\mathcal{G}}^{(j)}(\omega)= i{\mathcal{G}}^{(j)}(\omega)(1+\varepsilon^{(j)}(\omega)).
\end{align}
Finally a direct evaluation of $i{\mathcal{G}}^{(j)}$ using $\tilde{\Gamma}^{(j)},\tilde{\Delta}^{(j)}$ and a state that does not commute with $H$ leads to a relative error $\xi^{(j)}$, defined through
\begin{align}\label{eq:G_err_loc}
    i\tilde{\mathcal{G}}^{(j)}(\omega)=\frac{-(\tilde{\Gamma}^{(j)})^{2}(1+\xi^{(j)})}{\tilde{\Gamma}^{(j)}\omega+\tilde{\Delta}^{(j)}+i\tilde{\mathcal{G}}^{(j+1)}(\omega)}.
\end{align}
Using \cref{eq:G_err_gen} and \cref{eq:G_err_loc} leads to the following relation by applying the recursion that defines the continued fraction \cref{eq:tilde_G}
\begin{align}\label{eq:epsilon_condition}
1+\varepsilon^{(j)}(\omega)&=\frac{(1+\gamma^{(j)})^{2}(1+\xi^{(j)})}{1+\beta^{(j)}(\omega)+g^{(j)}(\omega)(\varepsilon^{(j+1)}(\omega)-\beta^{(j)}(\omega))},
\end{align}
with $g^{(j)}(\omega)=\frac{i\mathcal{G}^{(j+1)}(\omega)}{\Gamma^{(j)}\omega+\Delta^{(j)}+i\mathcal{G}^{(j+1)}(\omega)}$ and $\beta^{(j)}=\frac{\Gamma^{(j)}\omega\gamma^{(j)}+\Delta^{(j)}\delta^{(j)}}{\Gamma^{(j)}\omega+\Delta^{(j)}}$. Now we can state:
\begin{theorem}{Stability in evaluating continued fractions, Jones and Thron \cite{Jones1974}}

For each $k=1\dots n$, assume $\varepsilon^{(j)}(\omega)$ satisfies \cref{eq:epsilon_condition} with $g^{(n)}(\omega)=\varepsilon^{(n+1)}=0$. Additionally, let the non-negative number $s$ be chosen such that $\forall k\in 1\dots n$,  $|\delta^{(j)}|\leq s$, $|\beta^{(j)}(\omega)|\leq  s$, $|\xi^{(j)}|\leq s$, $|2\gamma^{(j)}+(\gamma^{(j)})^2|\leq s$ and $g^{(j)}(\omega)\leq q\leq 1$, then the relative error associated with the desired continued fraction $\mathcal{G}^{(0)}(\omega)$ satisfies
\begin{align}
|\varepsilon^{(0)}|\leq \frac{6s}{1-q},
\end{align}
and the algorithm is stable.
\end{theorem}
Theorem 4.1 in \cite{Jones1974} indicates that it is always possible to choose the imaginary part of $\omega$ large enough such that $g^{(j)}(\omega)$ satisfies the required condition. Note that, assuming a perfect quantum algorithm (i.e. no hardware and sampling noise so $\gamma^{(j)}=\delta^{(j)}=\beta^{(j)}=0$), but an initial state with infidelity $\epsilon$, this theorem tells us that the error of the algorithm is controlled by the infidelity, as the parameter $s$ is assumed to be upper-bounding $\xi^{(j)}$ which becomes equal to the error in the continued fraction at level $j$, $\varepsilon^{(j)}$ in this scenario.

We also investigate numerically this stability. In particular, to study the stability of the algorithm with respect to a state that does not exactly commute with the Hamiltonian, we use $\rho_k$ that is prepared with a $k$-layer VQE circuit as in \cite{stanisic2022observing}. The circuit diagram of the ansatz is shown in \cref{fig:VQE_circuit}.
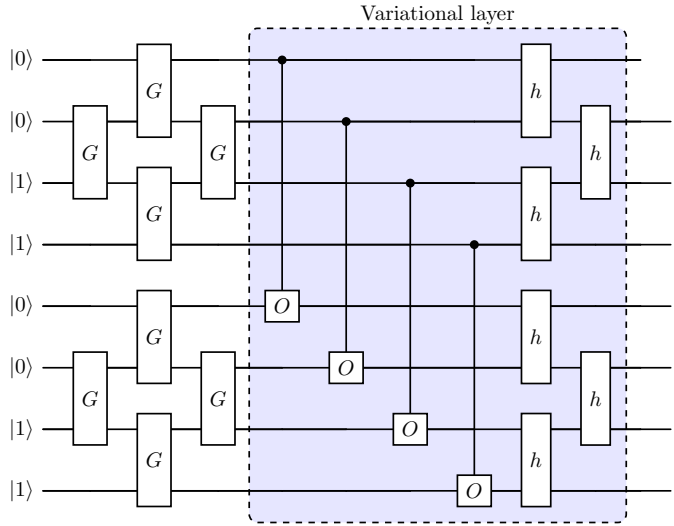
\begin{figure}[ht!]
\centering
\begin{adjustbox}{width=0.5\textwidth}
\begin{quantikz}
\lstick{$\ket{0}$}& \qw & \gate[2]{G} & \qw & \ctrl{4}\gategroup[8,steps=6,style = {dashed, rounded corners, fill = blue!10}, background]{Variational layer} & \qw & \qw & \qw & \gate[2]{h} & \qw & \qw \\
\lstick{$\ket{0}$} & \gate[2]{G} & \qw      & \gate[2]{G} & \qw & \ctrl{4} & \qw & \qw & \qw  &  \gate[2]{h} & \qw & \\
\lstick{$\ket{1}$} & \qw & \gate[2]{G} & \qw         & \qw      & \qw  & \ctrl{4} & \qw & \gate[2]{h} & \qw & \qw & \\
\lstick{$\ket{1}$} & \qw & \qw         & \qw         & \qw      & \qw & \qw & \ctrl{4} & \qw &  \qw & \qw & \\
\lstick{$\ket{0}$}& \qw & \gate[2]{G} & \qw        & \gate{O} & \qw & \qw & \qw & \gate[2]{h}  & \qw & \qw & \\
\lstick{$\ket{0}$} & \gate[2]{G} & \qw      & \gate[2]{G} & \qw & \gate{O} & \qw & \qw & \qw & \gate[2]{h} & \qw & \\
\lstick{$\ket{1}$} & \qw & \gate[2]{G} & \qw         & \qw      & \qw & \gate{O} & \qw & \gate[2]{h} & \qw & \qw & \\
\lstick{$\ket{1}$} & \qw & \qw         & \qw         & \qw      & \qw & \qw & \gate{O} & \qw & \qw & \qw & 
\end{quantikz}
\end{adjustbox}
\caption{Quantum circuit showing an instance of a VQE ansatz for a $N = 4$ site Fermi-Hubbard model at half-filling with one variational layer. We denote by $G$: Givens rotations; $O$: onsite (interaction) gates; $h$: hopping gates. Onsite and hopping gates correspond to time evolution according to the onsite and hopping terms, respectively. All onsite terms have the same time parameter, and all hopping terms occurring in parallel have the same time parameter. The first four qubits represent spin-up modes, and the last four represent spin-down modes. More details on the ansatz are given in \cite{stanisic2022observing}.}
\label{fig:VQE_circuit}
\end{figure}

 With this state, we evaluate the MBGF through the continued fraction algorithm, starting from the operator set $\{A_j\}_{j=0}^{N-1} = \{a_{j \uparrow}\}$ as shown in \cref{fig:VQE_a} (see also \cref{fig:VQE_n}). In this case, we observe a reasonable agreement with the exact behaviour for shallow VQE levels. The agreement improves as the fidelity increases, and the VQE becomes a better representation of the true ground state. It remains an open question of how fidelity has to scale with the system size to obtain a good agreement.

\begin{figure}[ht!]
    \centering
    \includegraphics[width = \linewidth]{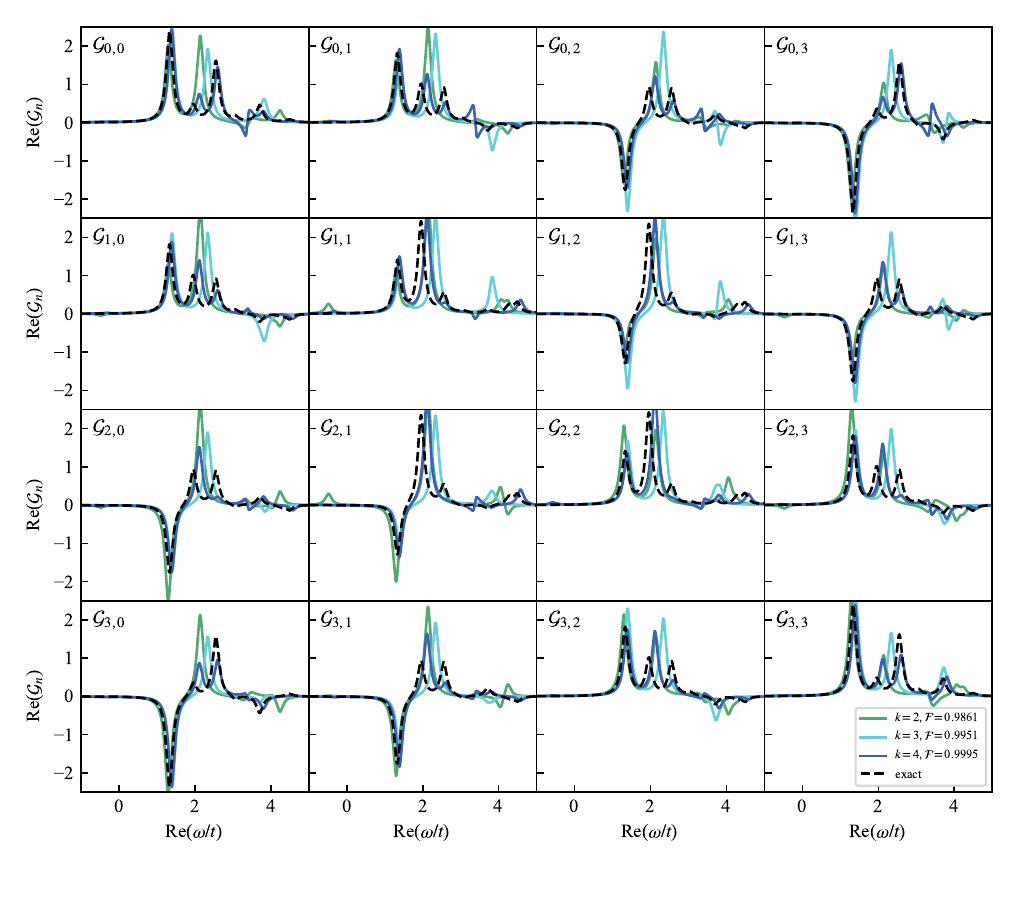}
    \caption{Approximation of the real part of the Fourier transformed MBGF obtained with the continued fraction algorithm for the 4-site Fermi-Hubbard model with $U/t = 4$ with the initial operator choice $\{A_j\}_{j=0}^{N-1} = \{a_{j \uparrow}\}$ and the ground state corresponding to the best VQE state of depth $k$ (\cref{fig:VQE_circuit}). We plot the results for $\eta = 0.1$ and compare them with the MBGF obtained using exact diagonalization. For each VQE depth $k$, we display the associated fidelity  $\mathcal{F}$ compared to the exact ground state. When performing the analysis, we restrict the $\bm{R}^{(j)}$ and $\bm{\Delta}^{(j)}$ matrices to being real.}
    \label{fig:VQE_a}
\end{figure}

In order to assess the effect of shot noise and algebraic manipulation of the noisy data in the many-operator based algorithm, we simulate the MBGF obtained from the sets $S=\{a_{0\uparrow},a_{1\uparrow},a_{2\uparrow}\}$ and truncation level $n=0$, (so just a single level of the matrix continued fraction) and $S'=\{a_{0\uparrow},a_{1\uparrow}\}$ with $n=1$, i.e. two levels of the iteration. Each iteration in this algorithm requires the inversion of the matrix of overlaps $\bm{R}$, defined in \cref{sec:algorithm}. This inversion process could be sensitive to the error from the finite number of samples, as has been found in other studies \cite{Epperly2022,lee2023sampling}. \cref{fig:shot_noise} illustrates this trend, as more iterations generate larger error as a function of the real part of the frequency. It is worth noting that a rather large number of measurements are required to obtain a small error. This indicates that for successful use of the continued fraction algorithm with real quantum data, it may be necessary to incorporate some thresholding in the matrix $\bm{R}$ to minimize errors coming from an ill-defined classical matrix of overlaps. We do not explore this question further in this work.

\begin{figure}[ht!]
    \centering
    \includegraphics[width = \linewidth]{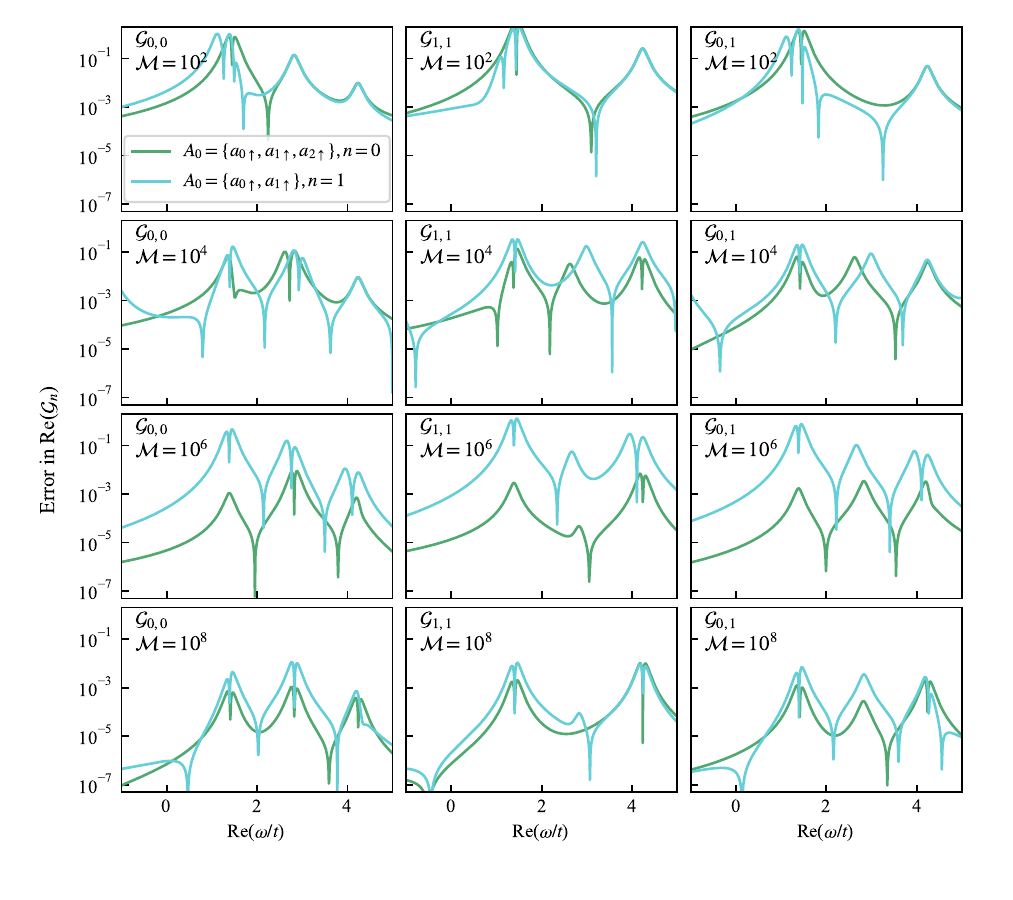}
    \caption{Approximation error of the real part of the MBGF for the $N=3$ site Fermi-Hubbard model with $U/t=2$ for 2 cases - $A_0 = \{a_{0 \uparrow}, a_{1 \uparrow}, a_{2 \uparrow}\}$ and $n=0$, and $A_0 = \{a_{0 \uparrow}, a_{1 \uparrow}\}$ and $n=1$. In both cases, the level $n$ is sufficient to obtain the exact solution. We investigate the response of the error to the statistical noise arising from the finite number of measurements $\mathcal{M}$ per each observable in the protocol.}
    \label{fig:shot_noise}
\end{figure}

\section{Discussion}\label{sec:Discussion}

The protocol that we discussed in this work, both for the single and many-operator case, can be implemented with access to a quantum computer by first classically obtaining the necessary observable values to be estimated at the state of interest, measuring them on the quantum computer, and classically post-process them to estimate the operators at the next stage of the continued fraction iteration. This process is then repeated $n$ times to reach the desired accuracy. This procedure does not require extra ancillas or controlled operations.

The results presented in the previous sections show that the continued fraction approximation to the MBGF converges exponentially in a region of the complex frequency above the real axis. Although we can prove this, the bounds we obtain are probably not tight. However, even considering a small imaginary part $\eta=0.1$ (in terms of the norm of the Hamiltonian) in \cref{sec:perfect}, we are able to reproduce the exact signal. It remains a question for future research if the approximation bounds presented here can be tightened.

It could be entirely possible, though, that in particular systems, an approximation based on a finite number of levels would only represent a good approximation to the real MBGF for a value of $\Im(\omega)$ that washes away the distinctive features of the signal. It is clear that if the resolution required is very high, for example, to try to resolve individual eigenvalues in a system where the gaps are closing exponentially fast with the system size, either an exponentially small $\Im(\omega)$ may be necessary or a large number of iterations. Nevertheless, we think this does not present an actual problem in real physical situations, where exponential resolution cannot be obtained anyhow on any experiment running during polynomial time. Moreover, based on extensive numerical evidence \cite{viswanath1994recursion}, it has recently been hypothesized \cite{Parker2019} that the generic growth of the $\Gamma^{(n)}$ coefficients of the single operator approach is linear with $n$, for generic Hamiltonians $H$ and operator $A$ in the thermodynamic limit. Under this hypothesis, the value of $\Lambda$ in \cref{theo:approx_corr} is finite for large $n$.

Among the desired properties of the algorithm is that by being a composition of M\"{o}bius maps with real coefficients, any truncation preserves the Nevalinna structure to the exact MBGF. This preserves the interpretation of the real part of the diagonal MBGF ($\Re{\mathcal{G}_{A_jA^\dagger_j}(\omega)}$) as density of states versus frequency of the excitation $A_j$ in the state $\rho$.

We showed, based on classical results on scalar continued fractions, that the single-operator based algorithm is stable to small perturbations of the constants $\Delta^{(j)},\Gamma^{(j)}$, naturally associated with sampling noise, and errors in the computation of the continued fraction, that are more closely related with defective state preparation. Through numerical simulations on small systems, we see that this stability could be achieved in the many-operator case but for either rather large fidelity with the true ground state \cref{fig:VQE_a}, or with a large number of shots \cref{fig:shot_noise} (see also \cref{fig:VQE_n} and \cref{fig:app_densities}). 
The error that stems from an approximation to a state with $[H,\rho]=0$ is not surprising, as the whole Hilbert space foliation that defines the algorithm is obtained from the inner product constructed from the state $\rho$, and the hermiticity of the Liouvillian with respect to that state

The sensitivity of the errors in the estimates of the overlaps between operators seems to be a problem in subspace expansion algorithms and is not unique to our algorithm \cite{Epperly2022}. In usual subspace expansion approaches, this error can be controlled through appropriate regularization of the overlap matrix \cite{Epperly2022,lee2023sampling}. In the overlap matrix, we choose to neglect and project out the eigenvalues of the overlap matrix that are less than some thresholding parameter $\epsilon$. This is a similar technique to what is used in quantum Krylov subspace expansion \cite{Kirby24}. The exact formulation of the errors arising in our context is outside the scope of this work and we leave a more quantitative investigation for the future.

Practically, the main cost of executing the algorithm under the assumption of access to a good state $\rho$ comes from the number of measurements and the classical resources to compute the nested commutators.  In \cref{app:measurement}, we discuss the associated sampling cost in terms of both variants of the continued fraction algorithm. The sample complexity grows exponentially for a generic state with the subspace expansion order $n$, as each order requires the computation of a nested commutator that generates an operator with a larger operator weight. A bound on the sampling complexity can be obtained in this case using the results of classical shadows \cite{Huang_2020, Huang2021b, Zhao_2021}, where it is shown that the required number of samples to estimate an operator grows exponentially with its weight. In the important case of fermionic systems, where the operators to evaluate correspond to even products of fermions, a similar analysis can be done using the fermionic shadow protocol \cite{Zhao_2021}. In that case the problem reduces to evaluating $k$-RDMs, where $k$ increases linearly with the subspace expansion order $n$. In practice, symmetries of the system and the possitivity of certain coefficients in the subspace expansion protocol can also be exploited to reduce the measurement cost. For Gibbs states under some mild assumptions \cite{onorati2023efficient,onorati2023provably}, the expectation of Lipschitz observables \cite{De_Palma_2021} can be obtained efficiently in general for large temperatures, indicating MBGFs of states at large temperatures are more accessible. 

We envision some mitigation techniques that can be used to reduce the sampling overhead. As discussed around \cref{eq:linear_system}, it is possible to use the reality condition of the Pauli expectations to lower the number of operators to be measured. This could be further lowered by using the inherent symmetries of the system in question. Symmetries could also be used to mitigate errors, as the set of MBGFs transforms under the action of a symmetry in a way that only depends on the operators chosen, meaning that this transformation can be obtained in advance of a quantum experiment and applied classically over the retrieved data, symmetrising the noisy signal. 

Although in this work, we have concentrated on states that are projections onto the ground state, it is possible to shift the question to study MBGFs for thermal states of nonzero temperature. There is evidence that creating or approximating these states on a quantum computer is efficient \cite{chen2023local,onorati2023efficient}. It is worth highlighting the connection with non-zero temperature states and complex frequencies. A natural way where a complexification of the frequency appears in many-body physics is through the Matsubara frequencies $\omega_n={\pi Tn}$ in systems at equilibrium. Here $T$ is the temperature of the Gibbs state, and $n$ is an even (odd) integer for bosonic (fermionic) systems. In that context, the discrete frequencies appear along the purely imaginary axis, and the given thermal expectations are defined on a strip of size $1/T$ in the complex frequency plane. This connection suggests that the approximation of the MBGF presented here could be used to study systems in the imaginary time formalism for large {\it Matsubara} frequencies.

Finally, we can relax the constraint of not being able to do time evolution at all, allowing a circuit that approximates the time evolution of the operator up to some time $T$ defined by the allowed quantum resources. In this case, the simplest MBGF becomes
\begin{equation}
\mathcal{G}_{AA^\dagger}(\omega)=\int_{0}^{\infty}e^{i\omega t}{\rm Tr}(\rho A(t)A^{\dagger})=\int_{0}^{T}dte^{i\omega t}{\rm Tr}(\rho A(t)A^{\dagger})+\int_{T}^{\infty}dte^{i\omega t}{\rm Tr}(\rho A(t)A^{\dagger}).
\end{equation}
The first term on the right-hand side above can be obtained directly by measuring ${\rm Tr}(\rho A(t)A^{\dagger})$ from times $t\in[0,T]$ using e.g. the method described in \cite{Somma_2003}, which requires an extra ancilla and a controlled application of $A$ (which we assume unitary for simplicity) dependent on the state of the ancilla. Then, the time integral can be produced classically. The remaining term is dampened by the imaginary part of $\omega$ as
\begin{equation}\label{eq:extra_term_time}
    \int_{T}^{\infty}dte^{i\omega t}{\rm Tr}(\rho A(t)A^{\dagger})=e^{i\omega T}\int_{0}^{\infty}dxe^{i\omega x}{\rm Tr}(\rho U_{T}A(x)U_{T}^{\dagger}A^{\dagger}),
\end{equation}
where $U_T=e^{iHT}$. The last integral in \cref{eq:extra_term_time} could be obtained with the many-operator based algorithm discussed in this work, with the operator set $S=\{U_TAU^\dagger_T,A\}$. Note that assuming the ability to
estimate ${\rm Tr}(\rho A(t)A^{\dagger})$ directly (e.g. by using ancillas as in \cite{Somma_2003}), then the difficulty of computing $\Gamma^{(n)},\Delta^{(n)}$ for the continued fraction approach will come from the number of circuits needed, and not by their weight. This contrasts with the complexity discussed in this work, which is based on the weight of the operators increasing as the continued fraction level grows (see also \cref{app:measurement}). The difference boils down to the assumption about the way of estimating observables. We leave the investigation of the feasibility of this approach as an avenue for future work.
\section{Conclusions}\label{sec:Conclusions}

Determining the MBGF, $\mathcal{G}_{AB}(\omega)=\int_0^\infty e^{i\omega t} {\rm Tr}(\rho e^{iHt}Ae^{-iHt}B)dt$
with continued fractions has a long history in physics, starting with the seminal work of Mori \cite{Mori1965}. Our main contribution in this work is to prove that finite approximants provide good approximations to the exact MBGF in the complexified frequency space. A motivation for studying this problem is the current difficulty of producing accurate deep circuits in quantum computers due to the effect of noise that cannot be fully corrected.

We have studied a subspace expansion algorithm based on the repeated application of the Liouvilian on a set of operators of interest. This repeated operation defines a foliation of the Hilbert space into orthogonal spaces under the inner product characterized by the state $\rho$. Starting with a single operator,  we revisit the continued fraction representation of the MBGF $\mathcal{G}_{AA^\dagger}$. Truncating to a finite level allows to calculate the MBGF function with constant depth (independent of the evolution time) quantum circuits, giving provable guarantees to the quality of the approximation. This is achieved by using the theory of orthogonal polynomials in the upper complex plane and the fact that the zeroes of these polynomials occur in the real axis. In the complex frequency space, it is then always possible to choose the imaginary part of the frequency such that the method converges as a Cauchy series, thus giving approximants that approach the infinite-level continued fraction exponentially fast in the level. We also extend this result to MBGF based on many operators, where a similar structure appears, but now in terms of matrix orthogonal polynomials, which are defined in terms of a three-term recursion with matrix coefficients.

The theoretical bounds presented here are further examined through experiments on a quantum computer and numerical simulations of the one-dimensional Fermi Hubbard model. In \cref{sec:perfect} we show that a perfect application of the algorithm renders the expected results with surprisingly low values of $\Im(\omega)$, for correlations functions of creation (annihilation) operators \cref{fig:fig1_u4} and \cref{fig:fig1_u2}.

We experimentally verified the correct application of the algorithm using a compressed version of the Hubbard dimer on 'IBM-Jakarta', \cref{fig:IBMQ}, where we study the MBGF using the number operator $n_{0\uparrow}=a_{0\uparrow}^\dagger a_{0\uparrow}$. The stability of the algorithm to imperfect preparation of the state $\rho$ and sample noise is discussed by adapting the classical result of \cite{Jones1974} of stability to small perturbations of the coefficients that define the continued fraction. The stability of the many-operator based algorithm to imperfect state preparation is evaluated in numerical simulation using $\rho$ obtained by VQE with the Hamiltonian variational ansatz, at different depths, with increasing fidelity \cref{fig:VQE_a}. Here, we observe that the error decreases with increasing fidelity.

The stability to sampling noise in subspace expansion algorithms has received attention lately \cite{Epperly2022, lee2023sampling} as it could lead to a failure of the algorithm due to the inversion of the overlap matrix. In the instance that we discussed in this paper, this inversion appears due to the normalization of the operators at each iteration. We investigate the effect of finite measurements by comparing the error obtained with different operator sets in \cref{fig:shot_noise}. There, we observe that for the same number of measurements, the instance of the algorithm with fewer inversions achieves better precision. This finding is consistent with the idea that inversions are detrimental to the algorithm's performance. A theoretical understanding of this phenomenon in the context of the present work, for the algorithm based on many operators, is left for a future study.

Implementing the algorithm presented in \cref{sec:algorithm} using a quantum computer requires the measurement of operators with weight increasing with the truncation level. Although expanding these operators in some operator basis could lead to an exponential explosion of terms as the continued fraction level increases, for local Hamiltonians the maximum weight of the operators in the sum would only grow linearly with the level. The sampling complexity for such a task grows exponentially with the weight of the operator involved but only logarithmically with the number of these operators. This complexity, on the other hand, leads to an exponentially better approximation away from the $\Im(\omega)=0$ line.

 Overall, the algorithm discussed here, in both its single and many-operator versions, provides an estimation of MBGF with shallow quantum circuits with provable guarantees.

\section*{Acknowledgements}

R.S. thanks Sabrina Wang, Filippo Gambetta, Evan Sheridan, Charles Derby and Laura Clinton for helpful discussions related to the MBGFs, subspace expansions and their use in determining many-body dynamics. 

\appendix
\appendix
\section*{APPENDIX}

\section{M\"{o}bius transformations and analytic structure of MBGFs}
\label{app:Mobius_Nevalina}

A M\"{o}bius transformation in the complex $z$ plane is a rational function 
\begin{align}
f(z)=\frac{az+b}{cz+d},
\end{align}
where $a,b,c,d$ are complex numbers satisfying $(ad-bc)\neq 0$. This function maps the upper complex plane ($\Im(z)>0$) to itself iff $(ad-bc)>0$ and $a,b,c,d$ are real numbers. This follows simply from imposing that $\Im(f(z))>0$ iff $\Im(z)>0$. This means that the map $s_n^\omega(z)$ in \cref{eq:Mobius} maps the upper complex plane into itself. 
\subsection{Nevalinna structure of the MBGFs}
A Nevanlinna function is a complex function which is analytic on the open upper half-plane $h$ and has a non-negative imaginary part. As discussed in \cref{sec:structure}, through the Lehmann representation we find that the MBGFs has two properties:
\begin{itemize}
\item The MBGFS $i\mathcal{G}_{AA^\dagger}(z+i\eta)$ has poles in the lower complex plane (as $\eta>0$)
\item $i\mathcal{G}_{AA^\dagger}(z+i\eta)$ (note the imaginary unit) has a positive imaginary part (see \cref{eq:Density_of_states}). 
\end{itemize}
These two properties make the MBGF $i\mathcal{G}_{AA^\dagger}(z+i\eta)$ a Nevalinna function.

\subsection{Sokhotski–Plemelj theorem}
The Sokhotski–Plemelj theorem relates the limiting values of the Cauchy integral in the complex plane $\zeta$
\begin{align}
\phi (z)=\frac{1}{2\pi i}\int_{C}{\frac {\varphi (\zeta )\,d\zeta }{\zeta -z}},
\end{align}
as $z$ approaches points in the closed curve $C$. It reads
\begin{align}
    \lim _{w\to z}\phi _{i}(w)={\frac {1}{2\pi i}}{\mathcal {P}}\int _{C}{\frac {\varphi (\zeta )\,d\zeta }{\zeta -z}}+{\frac {1}{2}}\varphi (z),\\
    \lim _{w\to z}\phi _{o}(w)={\frac {1}{2\pi i}}{\mathcal {P}}\int _{C}{\frac {\varphi (\zeta )\,d\zeta }{\zeta -z}}-{\frac {1}{2}}\varphi (z),
\end{align}
where the subindices $i,e$ denote the limit of approaching the curve from the inside or from the outside. Here $\mathcal{P}$ is the Cauchy principal value. In physics, in the context of Green's functions these relations are called Kramers-Kronig relations \cite{fetter2012quantum}. When the curve $C$ corresponds to the real axis, the curve is considered closed in the Riemann sphere.

\section{GNS construction}\label{app:GNS}

Here we discuss the idea behind the GNS construction. We aim to keep this discussion simple and intuitive. For formal results please refer to \cite{GelNeu43,Segal,doran1994c}.
 We want to define an inner product between operators based on some density matrix. Doing this directly could be problematic, as the density matrix could be such that it does not define a positive definite map. Consider for example a state $\rho$ in a given fixed sector of some symmetry. Operators that do not live in the same symmetry sector will have zero overlap with $\rho$. This means that the map defined by $\rho$ as
\begin{align}\label{eq:map}
    (A,B)_\rho:={\rm Tr}(\rho AB^\dagger),
\end{align}
does not define a norm, as $(A,A)$ may vanish for some operator different than the trivial $A=0$. The initial Hilbert space and the map \cref{eq:map} define a $C^*$ algebra. Let's call it $\mathcal{A}$ . The GNS construction is a way of constructing a proper Hilbert space from $\mathcal{A}$ and the map \cref{eq:map}. The main idea is that we can define a quotient space $\mathcal{H}=\mathcal{A}/{\mathcal{S}}$ between the original $C^*$ algebra and the subspace $\mathcal{S}\subset \mathcal{A}$ of $s\in \mathcal{S}$ operators such that
$(s,s)_\rho=0$. The inner product in $\mathcal{H}$ is defined using \cref{eq:map}, but over equivalence classes, where all elements in $\mathcal{S}$ correspond to the same equivalence class as the trivial $s=0$ element. In the works of Gelfand, Naimark and Segal, this construction is introduced and it is proved that it does make $\mathcal{H}$ a Hilbert space.

\section{Single operator approach}\label{app:single_op}

In \cref{sec:Algo}, we introduced the general many-operator approach to calculate the MBGF, starting from a set of initial operators $S_0$. We acknowledge that this general discussion is rather technical. To help the reader, we discuss here a special case of the more general approach, where the initial operator set $S_0$ corresponds to a single operator $A$. This appendix includes a simplified version of proof from \cref{app:cont_frac_many} for this special case.
In \cref{app:single_op_foliation}, we show the Hilbert space foliation generated by a single operator $A$, followed by a proof of the continued fraction approximation to the single operator Green's function in \cref{app:single_op_cf}.

\subsection{Foliation from a single operator $A$}\label{app:single_op_foliation}

A time-evolved operator $A(t)=e^{iHt}Ae^{-iHt}$ in the Heisenberg representation satisfies
$\frac{dA(t)}{dt}=i\mathcal{L}_HA(t)$.
This operator can be split in the components parallel to the subspace containing $A$ and its complement using the projection operator 
\begin{equation}
    P_A(X)=\frac{{\rm Tr}(\rho XA^\dagger)}{{\rm Tr}(\rho AA^\dagger)}A,
\end{equation}
which satisfies $P^2_A(X)=P_A(X)$. This projector is Hermitian with respect with the inner product $(X,Y)={\rm Tr}(\rho XY^\dagger)$ as
\be
(P_A(X),Y)={\rm Tr}\left(\frac{{\rm Tr}(\rho X A^\dagger)}{{\rm Tr}(\rho A A^\dagger)}\rho AY^\dagger\right)={\rm Tr}\left(\rho X A^\dagger\frac{{\rm Tr}(\rho AY^\dagger)}{{\rm Tr}(\rho A A^\dagger)}\right)={\rm Tr}(\rho X (P_A(Y))^\dagger)=(X,P_A(Y)).
\ee
The projector on the complement of $A$ is $Q_A = 1-P_A$. The evolution equation for the orthogonal component $Q_{A}(A(t))$ is
\begin{eqnarray}
\frac{d}{dt}Q_{A}(A(t))	=iQ_{A}\mathcal{L}_H(Q_{A}(A(t)))+i\frac{{\rm Tr}(\rho A(t)A^\dagger)}{{\rm Tr}(AA^\dagger)}Q_{A}\mathcal{L}_H(A),
\end{eqnarray}
which is simply obtained projecting the original evolution equation $\frac{dA(t)}{dt}=i\mathcal{L}_H A(t)$  in the space orthogonal to $A$.
This equation can be formally solved as
\be
Q_{A}A(t)=i\int_{0}^{t}d\tau\frac{{\rm Tr}(\rho A(\tau)A^\dagger)}{{\rm Tr}(AA^\dagger)} e^{i(t-\tau)Q_{A}\mathcal{L}_H}Q_{A}\mathcal{L}_H(A).
\ee
We see that to obtain the time evolution of operators in the orthogonal complement of $A$, we need the time evolution of the new operators $A^{(1)}(t)\equiv e^{itQ_{A}\mathcal{L}_H}Q_{A}(\mathcal{L}_H(A))$.
Note that the $t=0$ operator $A^{(1)}(0)=Q_{A}(\mathcal{L}_H(A))$ belongs to the orthogonal complement of $A$ and the operator $Q_{A}\mathcal{L}_H$ keeps the evolution in that orthogonal complement. We can repeat the previous discussion, now starting with the time evolution of $A^{(1)}(t)$.
Defining $Q_0=Q_A$ and $Q_1=Q_{A^{(1)}} $ we find
\be
\frac{d}{dt}Q_1 A^{(1)}(t)=iQ_1Q_0\mathcal{L}_H(Q_1 A^{(1)}(t))+i\frac{{\rm Tr}(\rho A^{(1)}(t)A^{(1)\dagger})}{{\rm Tr}(\rho A^{(1)}A^{(1)\dagger})}Q_1Q_0\mathcal{L}_H A^{(1)},
\ee
or equivalently
\be
Q_1 A^{(1)}(t)=i\int_{0}^{t}d\tau \frac{{\rm Tr}(\rho A^{(1)}(t)A^{(1)\dagger})}{{\rm Tr}(\rho A^{(1)}A^{(1)\dagger})} e^{i(t-\tau)Q_1Q_0\mathcal{L}_H}Q_1Q_{0}\mathcal{L}_H A^{(1)}.
\ee
Repeating this procedure $k-$times splits the Hilbert space into the orthogonal vector spaces spanned by operators $(A,A^{(1)},\dots,A^{(k)})$ and their complement, where
\be
A^{(j)}=\mathcal{L}_jA^{(j-1)}=Q_{j-1}Q_{j-2}\dots Q_{1}Q_{0}\mathcal{L}_H A^{(j-1)}, \quad Q_{j}=1-P_{j},
\ee
and $P_j=P_{A^{(j)}}$ ($A^{(0)}\equiv A$). These vector spaces define a foliation of the Hilbert space. The evolution in each subspace is determined by the projected Liouvilian operator
\be
\mathcal{L}_j=Q_{j-1}Q_{j-2}\dots Q_{1}Q_{0}\mathcal{L}_H,
\ee
meaning that $A^{(j)}(t)\equiv e^{it\mathcal{L}_j}A^{(j)}$.

\subsection{Continued fraction representation of Green's function}\label{app:single_op_cf}

Based on the previous foliation, Mori \cite{Mori1965} constructed a continued fraction representation of the time correlation as follows. Inserting the equation for the $j-$level operators
\be
Q_jA^{(j)}(t)=i\int_0^t d\tau\frac{{\rm Tr}(\rho A^{(j)}(\tau)A^{(j)\dagger})}{{\rm Tr}(\rho A^{(j)}A^{(j)\dagger})}A^{(j+1)}(t-\tau)
\ee
in the equation for the evolution of operator $A^{(j)}(t)$, $\frac{d}{dt} A^{(j)}(t)=i\mathcal{L}_j A^{(j)}(t)$
\begin{align}
\frac{d}{dt} A^{(j)}(t)&=i\mathcal{L}_jP_j A^{(j)}(t)+i\mathcal{L}_jQ_j A^{(j)}(t)\\
&=i\frac{{\rm Tr}(\rho A^{(j)}(t)A^{(j)\dagger})}{{\rm Tr}(\rho A^{(j)}A^{(j)\dagger})}\mathcal{L}_j A^{(j)}-\int_0^t d\tau\frac{{\rm Tr}(\rho A^{(j)}(\tau)A^{(j)\dagger})}{{\rm Tr}(\rho A^{(j)}A^{(j)\dagger})}\mathcal{L}_j A^{(j+1)}(t-\tau),
\end{align}
and taking the trace with respect to $A^{(j)\dagger} \rho$  on both sides, we find
\begin{align}\label{eq:inner_prod}
&\frac{d}{dt} {\rm Tr}(\rho A^{(j)}(t)A^{(j)\dagger}) =i\frac{{\rm Tr}(\rho A^{(j)}(t)A^{(j)\dagger})}{{\rm Tr}(\rho A^{(j)}A^{(j)\dagger})} {\rm Tr}(\rho\mathcal{L}_j A^{(j)}A^{(j)\dagger})\\&-\int_0^t d\tau\frac{{\rm Tr}(\rho A^{(j)}(\tau)A^{(j)\dagger})}{{\rm Tr}(\rho A^{(j)}A^{(j)\dagger})} {\rm Tr}(\rho\mathcal{L}_jA^{(j+1)}(t-\tau)A^{(j)\dagger}).
\end{align}
Defining ${G}^{(j)}(t)={\rm Tr}(\rho A^{(j)}(t)A^{(j)\dagger})$ and the constants $\Delta^{(j)}={\rm Tr}(\rho\mathcal{L}_j A^{(j)}A^{(j)\dagger})$ and \newline $\Gamma^{(j)}={\rm Tr}(\rho A^{(j)}A^{(j)\dagger})$ \cref{eq:inner_prod} becomes
\begin{eqnarray}\label{eq:evol_G}
\frac{d}{dt}{G}^{(j)}(t)&=&i{G}^{(j)}(t)\frac{\Delta^{(j)}}{\Gamma^{(j)}}-\int_0^t d\tau  {G}^{(j)}(\tau)\frac{(\mathcal L_jA^{(j+1)}(t-\tau),A^{(j)})}{\Gamma^{(j)}}
\end{eqnarray}
The last term above can be rewritten using the hermiticity of the projection operators, as \newline $Q_{j-1}\dots Q_1Q_0=1-\sum_{m=0}^{j-1}P_m$ and $(P_m(X),Y)=(X,P_m(Y))$ so we have
\begin{eqnarray}
(\mathcal L_jA^{(j+1)}(t-\tau),A^{(j)})=\left(\mathcal L_HA^{(j+1)}(t-\tau),A^{(j)}\right),
\end{eqnarray}
So far, we have not constrained the state $\rho$ that defines the inner product. Assuming that $\rho$
is such that the Liouvilian $\mathcal{L}_H$ is Hermitian with respect to the inner product, we have finally
\begin{eqnarray}
\label{eq:Hermiticity_L}
(\mathcal L_jA^{(j+1)}(t-\tau),A^{(j)})=\left(A^{(j+1)}(t-\tau),\mathcal L_HA^{(j)}\right)=(A^{(j+1)}(t-\tau),A^{(j+1)})=G^{(j+1)}(t-\tau),
\end{eqnarray}
where the second equality follows because $A^{j+1}(t-\tau)$ is already orthogonal to the subspace $1- \sum_{m=0}^j P_m$. Using Eq.~\eqref{eq:Hermiticity_L} in Eq.~\eqref{eq:evol_G} leads to
\begin{eqnarray}\label{eq:evol_G2}
\frac{d}{dt}{G}^{(j)}(t)&=&i{G}^{(j)}(t)\frac{\Delta^{(j)}}{\Gamma^{(j)}}-\int_0^t d\tau  {G}^{(j)}(\tau)\frac{G^{(j+1)}(t-\tau)}{\Gamma^{(j)}}.
\end{eqnarray}
Introducing the Laplace transform (defined for ${\rm Im}(\omega)>0$) $\mathcal{G}^{(j)}(\omega)=\int_0^\infty dt e^{i\omega t}G^{(j)}(t)$
we finally arrive at
\begin{equation}\label{eq:cont_frac}
    i\mathcal{G}^{(j)}(\omega)=\frac{-(\Gamma^{(j)})^{2}}{\Gamma^{(j)}\omega+\Delta^{(j)}+i\mathcal{G}^{(j+1)}(\omega)},
\end{equation}
which corresponds to the continued fraction representation of $\mathcal{G}^{(j)}(\omega)$.
Note, however, that in the continued fraction formalism, the weight $w$ of the operators $\{\Gamma^{(j)}, \Delta^{(j)}\}$ increases with the continued fraction iteration, thus the number of measurements that need to be performed scale exponentially with the number of continued fraction iterations but remains constant with the system size. Furthermore, since the continued fraction approximation has an error that decreases exponentially with the number of iterations, we argue that the algorithm is efficient in terms of measurements.
\section{Many operator approach}

In this appendix, we prove \cref{theo:CFM} in the main text for approximating the MBGFs using the matrix continued fraction expansion. 

\subsection{Proof of \cref{theo:CFM}} \label{app:cont_frac_many}

Based on the Hilbert space foliation introduced in \cref{sec:CFM}, a continued fraction representation of $\bm{\mathcal{G}}_{ij}(\omega)$ can be constructed as follows. Inserting the equation for the $j-$level operators
\begin{align}
Q_{j}A_{p}^{(j)}(t)&=i\sum_{k}\int_{0}^{t}d\tau{{\rm Tr}(\rho A_{p}^{(j)}(\tau)\hat{A}_{k}^{(j)\dagger})}e^{i(t-\tau)\mathcal{L}_{j+1}}\mathcal{L}_{j+1}\hat{A}_k^{(j)}\\&=i\sum_{k}\int_{0}^{t}d\tau{{\rm Tr}(\rho A_{p}^{(j)}(\tau)\hat{A}_{k}^{(j)\dagger})}A_k^{(j+1)}(t-\tau),
\end{align}
in the evolution equation for $\hat{A}_p^{(j)}(t)$, $\frac{d}{dt}\hat{A}_p^{(j)}(t)=i\mathcal{L}_{j}\hat{A}_p^{(j)}(t)$ we find
\begin{equation}
\frac{d}{dt}\hat{A}_{p}^{(j)}(t)=i\sum_{k}{{\rm Tr}(\rho \hat{A}_{p}^{(j)}(t)\hat{A}_{k}^{(j)\dagger})}\mathcal{L}_{j}\hat{A}_{k}^{(j)}-\sum_{k}\int_{0}^{t}d\tau{{\rm Tr}(\rho \hat{A}_{p}^{(j)}(\tau)\hat{A}_{k}^{(j)\dagger})}\mathcal{L}_{j}A_{k}^{(j+1)}(t-\tau).
\end{equation}
 Taking the trace with respect to $\hat{A}_{m}^{(j)\dagger}\rho$ on
both sides leads to
\begin{align}\label{eq:deriv_1}
\frac{d}{dt}{\rm Tr}(\rho \hat{A}_{p}^{(j)}(t)\hat{A}_{m}^{(j)\dagger})&=i\sum_{k}{{\rm Tr}(\rho \hat{A}_{p}^{(j)}(t)\hat{A}_{k}^{(j)\dagger})}{\rm Tr}(\rho\mathcal{L}_{j}\hat{A}_{k}^{(j)}\hat{A}_{m}^{(j)\dagger})\nonumber\\&-\sum_{k}\int_{0}^{t}d\tau{{\rm Tr}(\rho \hat{A}_{p}^{(j)}(\tau)\hat{A}_{k}^{(j)\dagger})}{\rm Tr}(\rho\mathcal{L}_{j}A_{k}^{(j+1)}(t-\tau)\hat{A}_{m}^{(j)\dagger}).
\end{align}
 Defining $G_{pm}^{(j)}(t)={\rm Tr}(\rho \hat{A}_{p}^{(j)}(t)\hat{A}_{m}^{(j)\dagger})$,
 $\Delta_{km}^{(j)}={\rm Tr}(\rho\mathcal{L}_{j}\hat{A}_{k}^{(j)}\hat{A}_{m}^{(j)\dagger})$, \cref{eq:deriv_1} becomes 
\begin{equation}
\frac{d}{dt}G_{pm}^{(j)}(t)=i\sum_{k}G_{pk}^{(j)}(t)\Delta_{km}^{(j)}-\sum_{k}\int_{0}^{t}d\tau G^{(j)}_{pk}(\tau){\rm Tr}(\rho\mathcal{L}_{j}A_{k}^{(j+1)}(t-\tau)\hat{A}_{m}^{(j)\dagger}).\label{eq:evol_G_many}
\end{equation}

The last term above can be rewritten using the hermiticity of the projection operators, as $Q_{j-1}\dots Q_{1}Q_{0}=1-\sum_{m=0}^{j-1}P_{m}$ where $P_m=\sum_k P_{A_j^{(m)}}$ and $(P_{m}(X),Y)=(X,P_{m}(Y))$ so we have 
\[
(\mathcal{L}_{j}A_{k}^{(j+1)}(t-\tau),\hat{A}_{m}^{(j)})=\left(\mathcal{L}_{H}A_{k}^{(j+1)}(t-\tau),\hat{A}_{m}^{(j)}\right).
\]

So far, we have not constrained the state $\rho$ that defines the
inner product. Assuming that $\rho$ is such that the Liouvilian $\mathcal{L}_{H}$
is Hermitian with respect to the inner product, we have finally 
\begin{eqnarray}
(\mathcal{L}_{j}A_{k}^{(j+1)}(t-\tau),\hat{A}_{m}^{(j)})=\left(A_{k}^{(j+1)}(t-\tau),\mathcal{L}_{H}\hat{A}_{m}^{(j)}\right)=(A_{k}^{(j+1)}(t-\tau),A_{m}^{(j+1)})\\
=\sum_{rs}(U^{\dagger(j+1)}\sqrt{D^{(j+1)}})_{kr}G_{rs}^{(j+1)}(t-\tau)(\sqrt{D^{(j+1)}}U^{(j+1)})_{sm},\label{eq:Hermiticity_L_matrix}
\end{eqnarray}
 where the second equality follows because $A_{k}^{j+1}(t-\tau)$
is already orthogonal to the subspace \newline $1-\sum_{m=0}^{j}P_{m}$. Using Eq.~\eqref{eq:Hermiticity_L_matrix} in \cref{eq:evol_G_many} leads to 
\begin{equation}
\frac{d}{dt}G_{pm}^{(j)}(t)=i\sum_{k}G_{pk}^{(j)}(t)\Delta_{km}^{(j)}-\sum_{krs}\int_{0}^{t}d\tau G_{pk}^{(j)}(\tau){[U^{(j)\dagger}\sqrt{D^{(j+1)}}]_{kr}}G_{rs}^{(j+1)}(t-\tau)[\sqrt{D^{(j+1)}}U^{(j+1)}]_{sm}.
\label{eq:evol_G_final}
\end{equation}

Introducing the Laplace transform (defined for ${\rm Im}(\omega)>0$)
$\mathcal{G}_{ab}^{(j)}(\omega)=\int_{0}^{\infty}dte^{i\omega t}G_{ab}^{(j)}(t)$ we find

\begin{equation}
- \delta_{pm}-i\omega\mathcal{G}^{(j)}_{pm}(\omega) =i\sum_{k}\mathcal{G}_{pk}^{(j)}(\omega)\Delta_{km}^{(j)}-\sum_{krs} \mathcal{G}^{(j)}_{pk}(\omega)[U^{(j)\dagger}\sqrt{D^{(j+1)}}]_{kr}\mathcal{G}_{rs}^{(j+1)}(\omega)[\sqrt{D^{(j+1)}}U^{(j+1)}]_{sm},
\label{eq:evol_Glap}
\end{equation}
which in matrix form becomes
\begin{equation}
i\bm{\mathcal{G}}^{(j)}(\omega)=- \left[\omega+\bm{\Delta}^{(j)} +\bm{M}^{\dagger(j+1)}i\bm{\mathcal{G}}^{(j+1)}(\omega)\bm{M}^{(j+1)}\right]^{-1},
\label{eq:evol_G_matrix}
\end{equation}
with $\bm{M}^{(j+1)}:= \sqrt{\bm{D}^{(j+1)}}\bm{U}^{(j+1)}$. This is the continued fraction representation of the matrix $(\bm{\mathcal{G}}^{(j)}(\omega))_{ab}:=\mathcal{G}_{ab}^{(j)}(\omega)$.

The recursion relation found in \cref{eq:evol_G_matrix} defines the continued fraction representation of $\bm{\mathcal{G}}^{(0)}$ as the infinite iterated map
$\bm{\mathcal{G}}^{(0)}(\omega)= \lim_{k\rightarrow\infty} \bm{\mathcal{G}}_k^{(0)}(\omega),$ where 
\begin{equation}\label{eq:mat_cont_frac_mob}
    \bm{\mathcal{G}}^{(0)}_k(\omega)=\bm{s}^\omega_0\circ \bm{s}_1^\omega \circ \dots \circ \bm{s}_{k-1}^\omega\circ \bm{s}_k^\omega(0)\quad \mbox{and} \quad \bm{s}^\omega_j(\bm{X} ):=
- \left[\omega+\bm{\Delta}^{(j)} +\bm{M}^{\dagger(j+1)}\bm{X}\bm{M}^{(j+1)}\right]^{-1}.
\end{equation}
Let's define the map $T_{\bm{S}}(\bm{X}):=(\bm{A}\bm{X}+\bm{B})(\bm{C}\bm{X}+\bm{D})^{-1}$, where $\bm{S}$ is given by 
$
   \bm{S}= \left(\begin{array}{cc}
    \bm{A}& \bm{B}\\
    \bm{C} & \bm{D}
    \end{array}\right).
$
It is easy to check that this map satisfies the composition rule
\begin{equation}\label{eq:comp_rule}
T_{\bm{S}_1}\circ T_{\bm{S}_2}(\bm{X}) := T_{\bm{S}_1}(T_{\bm{S}_2}(\bm{X}))=T_{\bm{S}_1\bm{S}_2}(\bm{X}). 
\end{equation}
Noting that $\bm{s}^\omega_j=T_{\bm{S}_j(\omega)\bm{R}_{j+1}}$ with 

\begin{equation}
\bm{S}_j(\omega)=\left(\begin{array}{cc}
   0  & -1 \\
   1  & \omega + \bm{\Delta}^{(j)}.
\end{array}\right) \quad \text{and} \quad\bm{R}_{j+1}=\left(\begin{array}{cc}
   \bm{M}^{\dagger(j+1)}  & 0 \\
   0  & (\bm{M}^{(j+1)})^{-1}
\end{array}\right),
\end{equation}
the MBGF $\bm{\mathcal{G}}^{(0)}$ can be cast as
\begin{equation}
    \bm{\mathcal{G}}^{(0)}(\omega)= \lim_{k\rightarrow\infty} T_{\left(\prod_{j=0}^{k-1}\bm{S}_j(\omega)\bm{R}_{j+1}\right)\bm{S}_k(\omega)}(0).
\end{equation}

\section{Measurement Strategy}\label{app:measurement}

In this appendix, we discuss the scaling of the number of measurements as a function of the continued fraction level for calculating the MBGF. Let's concentrate on the single-operator case for simplicity. In the continued fraction algorithm, we can classically compute the commutator $[H, A^{(j)}]$ that is needed to estimate the coefficients $\{\Gamma^{(j)}, \Delta^{(j)}\}$. After mapping into Paulis, the initial operator $A$ has support on $k_A$ sites, and we have a Hamiltonian where each term has support on most $k_H$ sites. After a single level, the operator $[H,A]$ will have maximum support on $k_A+k_H$ sites. The number of operators in the Pauli operator basis will also, generically, grow exponentially with the level. Using an efficient measurement scheme in terms of the number of operators to measure \cite{Huang_2020,Huang2021b}, the main cost will depend exponentially on the Pauli weight $w$, which will grow linearly with the level. If further constraints are imposed on the state and the operators being measured, this cost can be further reduced \cite{onorati2023efficient}.

In the many-operator case, the situation is different, e.g. already at the 0-th iteration, we can have that $\Tr(\rho A^0_i A^{0\dagger}_j)$ is spatially separated at opposite ends of the system. In the case of fermionic observables, this can lead to support across the entire system when a fermionic encoding to qubits is used. In this case, a different measurement strategy based on the fermionic shadow approach \cite{Zhao_2021} may be advantageous. The approach is based on evaluating $k$-RDM of $\rho$, which can be represented by a $2k$ index tensor:
\be
    ^k D^{p_1 \dots p_k}_{q_1 \dots q_k} = \Tr(a^{\dagger}_{p_1}\dots a^\dagger_{p_k} a_{q_k}\dots a_{q_1} \rho).
\ee
These $k$-RDMs exactly correspond to the structure of measurements when probing MBGFs of fermionic models. Using fermionic shadows, one can estimate all $k$-RDMs up to additive error $\epsilon$ using 
$M = \mathcal{O}({n \choose k} \frac{k^{3/2} \log(n)}{\epsilon^2})$ copies of state $\rho$.

To evaluate the cost of this approach, let us consider operators $A, B$ with fermionic operator support on $k_A$ and $k_B$ modes, respectively (in the examples that we consider with annihilation operators, they would simply equal 1). Furthermore, consider a fermionic Hamiltonian $H$ where each term has at most $d$ different modes. From here, we can establish the following relations. 
\be 
    D_{A,B} = \Tr(\rho A B^\dagger)
\ee
is at most a $(k_A + k_B)$-RDM. Contrary to the previous case, where we were looking for the support of the operators, here we look for the number of fermionic modes featured. One can easily see that $\comm{H}{A}$ will have terms with at most $k_A + d$ different modes. This allows to establish that: 
\be
    D_{H, A, B} = \Tr(\rho \comm{H}{A} B^\dagger)
\ee
is at most $(d+k_A+k_B)$-RDM.

Consider starting with a set of operators $\{A_0^{j}\}_{j=1}^{N}$, with the maximum number of modes $k_0$. To reach the next level of continued fractions, we need at most $2k_0+d$-RDM.Thus, to reach $n^{th}$ level of continued fractions, we need to calculate 
at most $(2k_0+(2n+1)d)$-RDM. Thus, the total cost for the measurements is:
\be
    C = \mathcal{O}\left( \sum_{x=1}^{2k_0+(2n+1)d}{N \choose x}x^{3/2}\log(N)/\epsilon^2 \right).
\ee

Note that although the required number of measurements still scales exponentially with the continued fraction level, this cost is independent of the fermionic encoding used, and it is expected to scale better than the Pauli measurement for Hamiltonians acting on few different modes but with geometrically nonlocal interactions.

\section{Extra Results}\label{app:results_app}

\subsection{Additional results}

In this appendix we include additional results for the approximation of the MBGF and its error in \cref{fig:fig1_u2} for the $N=4$ site Fermi-Hubbard model, with the set of initial operators $\{A_j\}_{j=1}^N = \{a_{j \uparrow}\}$ and $U/t = 2$ (complementing the result in \cref{fig:fig1_u4} for $U/t = 4$). In both cases, we investigate the agreement between the approximation obtained at order $n$ with continued fractions and the respective exact MBGF. 

\begin{figure}
    \centering
    \includegraphics[width = \linewidth]{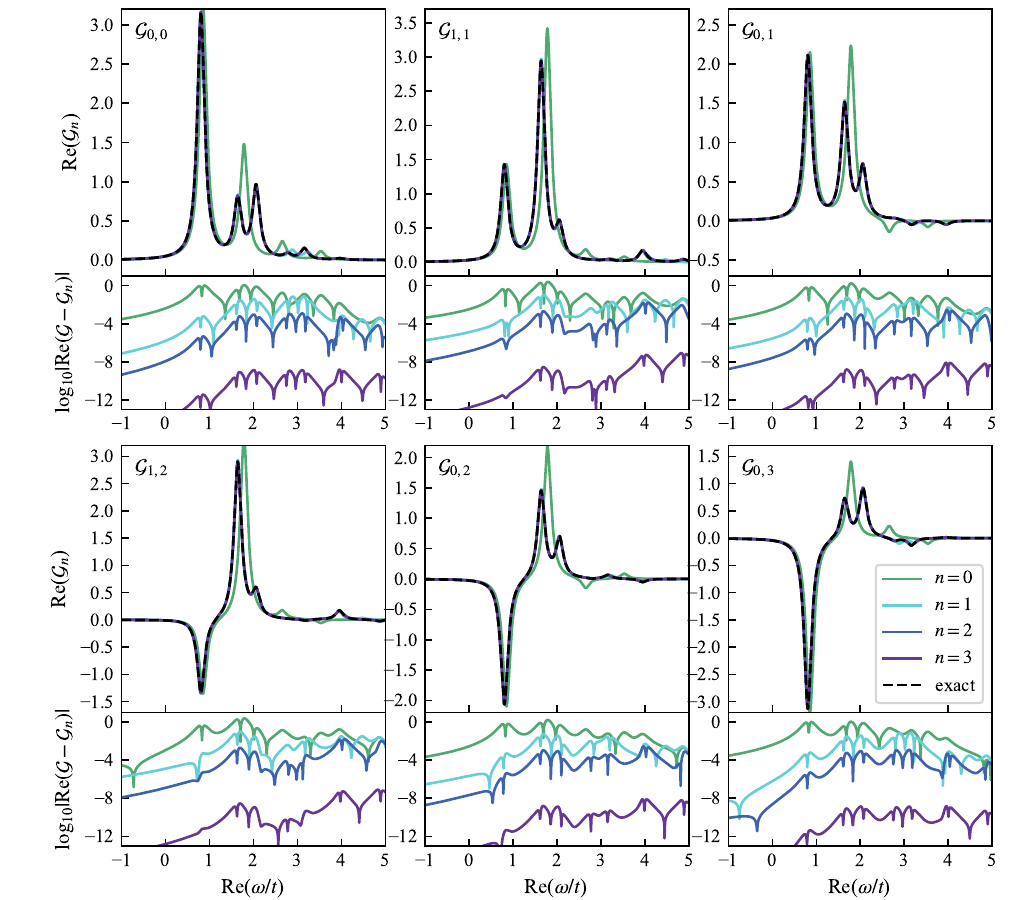}
    \caption{The approximation of the Fourier transformed Green's function obtained with the continued fraction algorithm at different levels $n$ for the 4-site Fermi-Hubbard model with $U/t = 2$ with the initial operator choice $S_0 = \{a_{j \uparrow}\}_{j=0}^{N-1}$. We plot the results for $\eta = 0.1$ and compare them with the MBGF obtained using exact diagonalization.}
    \label{fig:fig1_u2}
\end{figure}

In \cref{fig:VQE_n}, we evaluate the MBGF at the continued fraction level $n=3$ and for various VQE circuit depths $k$, based on the operator set $\{A_j\}_{j=1}^N = \{n_{j \uparrow}\}$ for $N=4$. One observes a good agreement with the exact MBGF for all VQE depth $k$ values.

\begin{figure}[ht!]
    \centering
    \includegraphics[width = \linewidth]{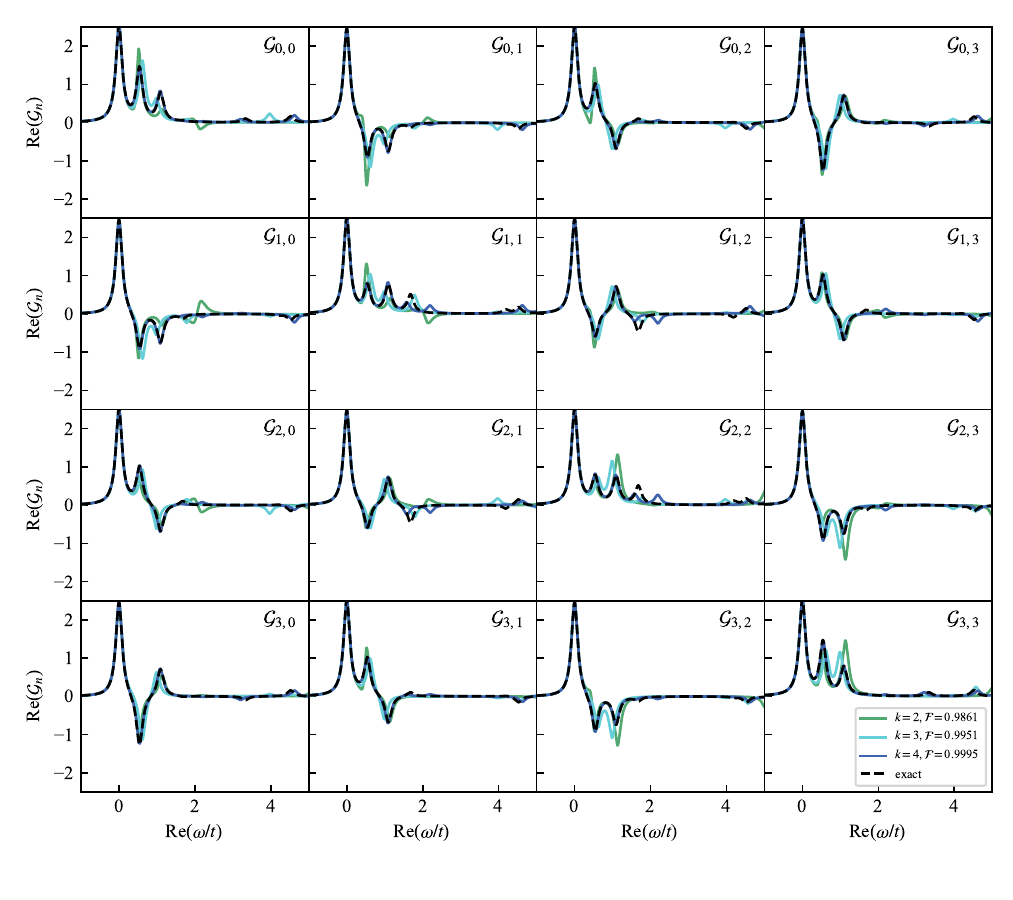}
    \caption{The approximation of the real part of the Fourier transformed Green's function obtained with the continued fraction algorithm for the 4-site Fermi-Hubbard model with $U/t = 4$ with the initial operator choice $S_0 = \{n_{j \uparrow}\}_{j=0}^{N-1}$ and the ground state corresponding to the best VQE state of depth $k$. We plot the results for $\eta = 0.1$ and compare them with the MBGF obtained using exact diagonalization. For each VQE depth $k$, we display the associated fidelity compared to the exact ground state. }
    \label{fig:VQE_n}
\end{figure}

\begin{figure}
    \centering
    \includegraphics[width = \linewidth]{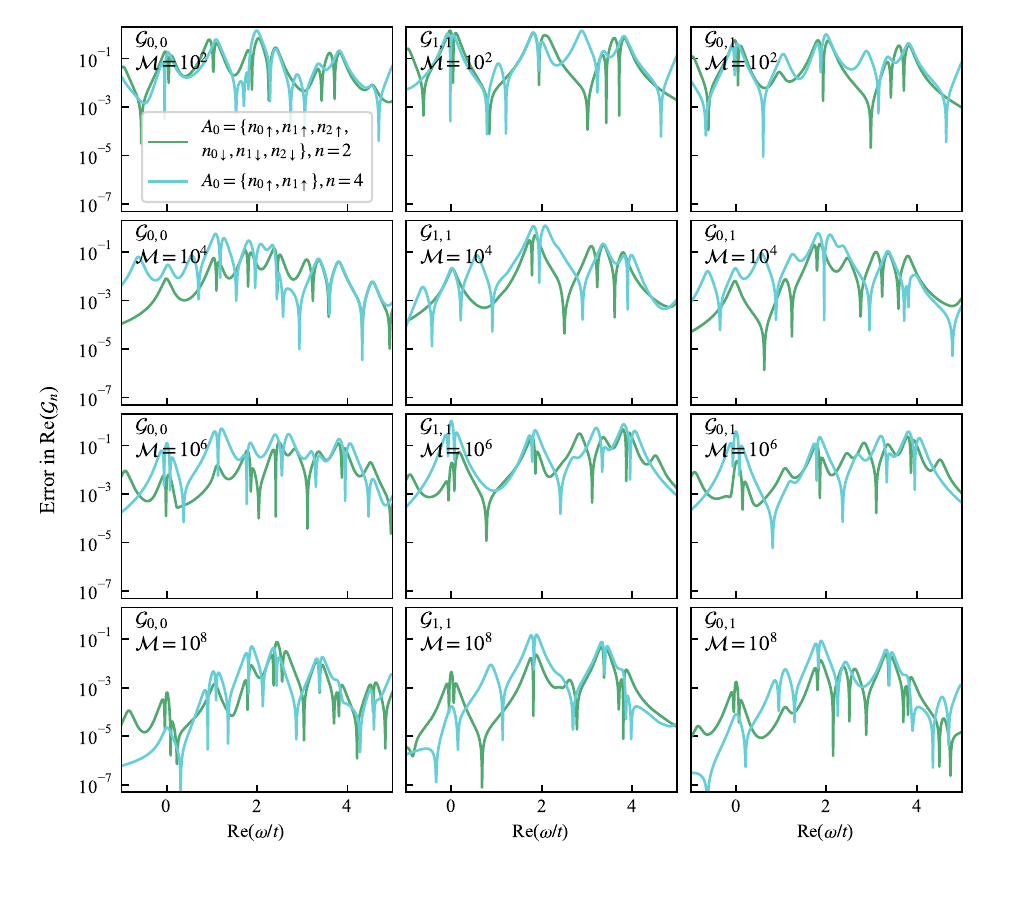}
    \caption{Approximation error of the $\Re{\mathcal{G}_{ab}}$ for the three site Fermi-Hubbard model with $U/t=2$. We choose two initial operator sets, $S=\{n_{i,\sigma}\}_{\forall i, \sigma}$ and $S'=\{n_{i\uparrow}\}_{i=0}^1$. The set $S$ is iterated two times $n=2$, while the set $S'$ is iterated 4 times $(n=4)$  In either case, the level $n$ is sufficient to obtain the exact solution. We investigate the response of the error to the statistical noise arising from the finite number of measurements $\mathcal{M}$ per each observable in the protocol. }
    \label{fig:app_densities}
\end{figure}

\subsection{Symmetry of the real part of $\mathcal{G}_{ab}(\omega)$}\label{sec:results_app}

In this section, we show that for any system commuting with an antiunitary operator $\mathcal T$, such that $\mathcal T^2=1$, the real part of the MBGF considered in this work is symmetric, i.e. $\Re(\calG_{ab}) = \Re(\calG_{ba})$. Assume $\rho = \sum_j \rho(E_j) \ket{E_j}\bra{E_j}$, where $|E_j\rangle$ are the eigenstates of $H$. Let's denote $c_{mj}^a:=\bra{E_m}A_a\ket{E_j}$ and split $\omega = \omega_0 + i\eta$ with $\eta>0$, then
\begin{align}\nonumber
    &\mathcal{G}_{ab}(\omega) = \int_0^{\infty} dt e^{i\omega t} \Tr(\rho A_a(t)A_b^{\dagger}) = \int_0^{\infty} dt \sum_{mj} \rho(E_m) e^{i(E_m-E_j+\omega)t}\bra{E_m}  A_a\ket{E_j}\bra{E_j}A_b^\dagger\ket{E_m},\\&=i\sum_{mj}\frac{\rho(E_m) c_{mj}^a c_{mj}^{b*}}{\omega +E_m -E_j}=\sum_{mj}\frac{\rho(E_m) c_{mj}^a c_{mj}^{b*}}{(\omega_0 +E_m -E_j)^2+\eta^2}(\eta+i(\omega_0 +E_m -E_j))
\end{align}
such that
\begin{align}
    \Re(\mathcal{G}_{ab}(\omega)) = \frac{\sum_{mj}\rho(E_m) }{(\omega_0 +E_m -E_j)^2+\eta^2}(\Re(c_{mj}^a c_{mj}^{b*})\eta-\Im(c_{mj}^a c_{mj}^{b*})(\omega_0 +E_m -E_j)).
\end{align}
To show that $\Re(\mathcal{G}_{ab}(\omega))) = \Re(\mathcal{G}_{ba}(\omega))$, we need to show that $\Im(c_{mj}^a c_{mj}^{b*}) = 0$. For $[H,\mathcal{T}]=0$, the non-degenerate eigenstates of $H$ satisfy $\mathcal{T}|E_j\rangle=e^{i\alpha}|E_j\rangle$. So it is always possible to choose the non-degenerate $|\tilde{E}_j\rangle$ such that $\mathcal{T}|\tilde{E}_j\rangle=|\tilde{E}_j\rangle$, by redefining $|\tilde{E}_j\rangle :=e^{i\alpha}|E_j\rangle$ and using the antiunitarity of $\mathcal{T}$. Any antiunitary operator can be expressed as complex conjugation $\mathcal{K}$ times a unitary operator. The model considered in this work commutes with $\mathcal{K}$ so, its eigenvectors can be chosen to be real. From this observation, it is evident that any coefficient $\bra{E_i}A\ket{E_j}$ is real for the operators $A = a_{i\sigma}$ or $A = n_{i\sigma}$ that we consider. This implies that $\Im(c_{mj}^ac_{mj}^{b*})=0$ making $\Re(\calG_{ab}) = \Re(\calG_{ba})$.

\bibliographystyle{quantum}
\bibliography{main}

\begin{thebibliography}{10}

\bibitem{kubo1957statistical}
Ryogo Kubo.
\newblock ``Statistical-mechanical theory of irreversible processes. i. general theory and simple applications to magnetic and conduction problems''.
\newblock \href{https://dx.doi.org/https://doi.org/10.1143/JPSJ.12.570}{Journal of the physical society of Japan {\bf 12}, 570--586}~(1957).

\bibitem{abrikosov1997locally}
I.~A. Abrikosov, S.~I. Simak, B.~Johansson, A.~V. Ruban, and H.~L. Skriver.
\newblock ``Locally self-consistent green's function approach to the electronic structure problem''.
\newblock \href{https://dx.doi.org/https://doi.org/10.1103/PhysRevB.56.9319}{Phys. Rev. B {\bf 56}, 9319--9334}~(1997).

\bibitem{fetter2012quantum}
Alexander~L Fetter and John~Dirk Walecka.
\newblock ``Quantum theory of many-particle systems''.
\newblock \href{https://dx.doi.org/https://doi.org/10.1063/1.3071096}{Courier Corporation}. ~(2012).

\bibitem{Kallen1952}
Gunnar Källén.
\newblock ``On the definition of the renormalization constants in quantum electrodynamics''.
\newblock \href{https://dx.doi.org/https://doi.org/10.5169/seals-112316}{Helvetica Physica Acta {\bf 25}, 417}~(1952).

\bibitem{lehmann1954eigenschaften}
Harry Lehmann.
\newblock ``{\"U}ber eigenschaften von ausbreitungsfunktionen und renormierungskonstanten quantisierter felder''.
\newblock \href{https://dx.doi.org/https://doi.org/10.1007/BF02783624}{Il Nuovo Cimento (1943-1954) {\bf 11}, 342--357}~(1954).

\bibitem{georges1996dynamical}
Antoine Georges, Gabriel Kotliar, Werner Krauth, and Marcelo~J. Rozenberg.
\newblock ``Dynamical mean-field theory of strongly correlated fermion systems and the limit of infinite dimensions''.
\newblock \href{https://dx.doi.org/https://doi.org/10.1103/RevModPhys.68.13}{Rev. Mod. Phys. {\bf 68}, 13--125}~(1996).

\bibitem{kotliar2006electronic}
G.~Kotliar, S.~Y. Savrasov, K.~Haule, V.~S. Oudovenko, O.~Parcollet, and C.~A. Marianetti.
\newblock ``Electronic structure calculations with dynamical mean-field theory''.
\newblock \href{https://dx.doi.org/https://doi.org/10.1103/RevModPhys.78.865}{Rev. Mod. Phys. {\bf 78}, 865--951}~(2006).

\bibitem{aoki2014nonequilibrium}
Hideo Aoki, Naoto Tsuji, Martin Eckstein, Marcus Kollar, Takashi Oka, and Philipp Werner.
\newblock ``Nonequilibrium dynamical mean-field theory and its applications''.
\newblock \href{https://dx.doi.org/https://doi.org/10.1103/RevModPhys.86.779}{Rev. Mod. Phys. {\bf 86}, 779--837}~(2014).

\bibitem{Bauer2016}
Bela Bauer, Dave Wecker, Andrew~J. Millis, Matthew~B. Hastings, and Matthias Troyer.
\newblock ``Hybrid quantum-classical approach to correlated materials''.
\newblock \href{https://dx.doi.org/https://doi.org/10.1103/PhysRevX.6.031045}{Phys. Rev. X {\bf 6}, 031045}~(2016).

\bibitem{Kreula2016}
J.~M. Kreula, S.~R. Clark, and D.~Jaksch.
\newblock ``{Non-linear quantum-classical scheme to simulate non-equilibrium strongly correlated fermionic many-body dynamics}''.
\newblock \href{https://dx.doi.org/https://doi.org/10.1038/srep32940}{Scientific Reports {\bf 6}, 1--7}~(2016).

\bibitem{Kosugi2020}
Taichi Kosugi and Yu~Ichiro Matsushita.
\newblock ``{Construction of Green's functions on a quantum computer: Quasiparticle spectra of molecules}''.
\newblock \href{https://dx.doi.org/https://doi.org/10.1103/PhysRevA.101.012330}{Phys. Rev. A {\bf 101}, 012330}~(2020).

\bibitem{Ciavarella2020}
Anthony Ciavarella.
\newblock ``Algorithm for quantum computation of particle decays''.
\newblock \href{https://dx.doi.org/https://doi.org/10.1103/PhysRevD.102.094505}{Phys. Rev. D {\bf 102}, 094505}~(2020).

\bibitem{Endo2020}
Suguru Endo, Iori Kurata, and Yuya~O. Nakagawa.
\newblock ``Calculation of the green's function on near-term quantum computers''.
\newblock \href{https://dx.doi.org/https://doi.org/10.1103/PhysRevResearch.2.033281}{Phys. Rev. Res. {\bf 2}, 033281}~(2020).

\bibitem{Chen2021}
Hongxiang Chen, Max Nusspickel, Jules Tilly, and George~H. Booth.
\newblock ``Variational quantum eigensolver for dynamic correlation functions''.
\newblock \href{https://dx.doi.org/https://doi.org/10.1103/physreva.104.032405}{Phys. Rev. A {\bf 104}, 032405}~(2021).

\bibitem{Cruz2022}
Diogo Cruz and Duarte Magano.
\newblock ``Superresolution of green's functions on noisy quantum computers''.
\newblock \href{https://dx.doi.org/https://doi.org/10.1103/PhysRevA.108.012618}{Phys. Rev. A {\bf 108}, 012618}~(2023).

\bibitem{Sakurai2022}
Rihito Sakurai, Wataru Mizukami, and Hiroshi Shinaoka.
\newblock ``{Hybrid quantum-classical algorithm for computing imaginary-time correlation functions}''.
\newblock \href{https://dx.doi.org/https://doi.org/10.1103/PhysRevResearch.4.023219}{Phys. Rev. Res. {\bf 4}, 023219}~(2022).

\bibitem{Libbi2022}
Francesco Libbi, Jacopo Rizzo, Francesco Tacchino, Nicola Marzari, and Ivano Tavernelli.
\newblock ``Effective calculation of the green's function in the time domain on near-term quantum processors''.
\newblock \href{https://dx.doi.org/https://doi.org/10.1103/PhysRevResearch.4.043038}{Phys. Rev. Res. {\bf 4}, 043038}~(2022).

\bibitem{Jamet2022a}
Francois Jamet, Abhishek Agarwal, and Ivan Rungger.
\newblock ``Quantum subspace expansion algorithm for green's functions''~(2022).
\newblock  \href{http://arxiv.org/abs/2205.00094}{arXiv:2205.00094}.

\bibitem{CQ_Lanczos}
Gabriel Greene-Diniz, David~Zsolt Manrique, Kentaro Yamamoto, Evgeny Plekhanov, Nathan Fitzpatrick, Michal Krompiec, Rei Sakuma, and David~Mu{\~{n}}oz Ramo.
\newblock ``Quantum {C}omputed {G}reen's {F}unctions using a {C}umulant {E}xpansion of the {L}anczos {M}ethod''.
\newblock \href{https://dx.doi.org/https://doi.org/10.22331/q-2024-06-20-1383}{{Quantum} {\bf 8}, 1383}~(2024).

\bibitem{Claudino_2021}
Daniel Claudino, Bo~Peng, Nicholas~P Bauman, Karol Kowalski, and Travis~S Humble.
\newblock ``Improving the accuracy and efficiency of quantum connected moments expansions''.
\newblock \href{https://dx.doi.org/https://doi.org/10.1088/2058-9565/ac0292}{Quantum Science and Technology {\bf 6}, 034012}~(2021).

\bibitem{viswanath1994recursion}
V.S. Viswanath and G.~M{\"u}ller.
\newblock ``The recursion method: Application to many body dynamics''.
\newblock \href{https://dx.doi.org/https://doi.org/10.1007/978-3-540-48651-0}{Springer Berlin Heidelberg}. ~(1994).

\bibitem{nandy2024quantum}
Pratik Nandy, Apollonas~S. Matsoukas-Roubeas, Pablo Martínez-Azcona, Anatoly Dymarsky, and Adolfo del Campo.
\newblock ``Quantum dynamics in krylov space: Methods and applications''~(2024).
\newblock  \href{http://arxiv.org/abs/2405.09628}{arXiv:2405.09628}.

\bibitem{nevanlinna1922asymptotische}
R.~Nevanlinna.
\newblock ``Asymptotische entwicklungen beschr{\"a}nkter funktionen und das stieltjessche momentenproblem''.
\newblock Annales Academiae scientiarum Fennicae. Suomalainen tiedeakatemia. ~(1922).

\bibitem{Luger2014}
Annemarie Luger.
\newblock ``Generalized nevanlinna functions: Operator representations, asymptotic behavior''.
\newblock \href{https://dx.doi.org/10.1007/978-3-0348-0692-3_35-2}{Pages 1--24}.
\newblock Springer Basel. Basel~(2014).

\bibitem{Huang_2020}
Hsin-Yuan Huang, Richard Kueng, and John Preskill.
\newblock ``Predicting many properties of a quantum system from very few measurements''.
\newblock \href{https://dx.doi.org/https://doi.org/10.1038/s41567-020-0932-7}{Nature Physics {\bf 16}, 1050–1057}~(2020).

\bibitem{Huang2021b}
Hsin-Yuan Huang, Richard Kueng, and John Preskill.
\newblock ``Efficient estimation of pauli observables by derandomization''.
\newblock \href{https://dx.doi.org/https://doi.org/10.1103/PhysRevLett.127.030503}{Phys. Rev. Lett. {\bf 127}, 030503}~(2021).

\bibitem{Zhao_2021}
Andrew Zhao, Nicholas~C. Rubin, and Akimasa Miyake.
\newblock ``Fermionic partial tomography via classical shadows''.
\newblock \href{https://dx.doi.org/https://doi.org/10.1103/physrevlett.127.110504}{Phys. Rev. Lett. {\bf 127}, 110504}~(2021).

\bibitem{Mori1965}
Hazime Mori.
\newblock ``{A Continued-Fraction Representation of the Time-Correlation Functions}''.
\newblock \href{https://dx.doi.org/https://doi.org/10.1143/PTP.34.399}{Progress of Theoretical Physics {\bf 34}, 399--416}~(1965).

\bibitem{GelNeu43}
I~Gelfand and M~Neumark.
\newblock ``On the imbedding of normed rings into the ring of operators in hilbert space''.
\newblock Rec. Math. [Mat. Sbornik] N.S. {\bf 54}, 197--217~(1943).
\newblock  url:~\url{http://eudml.org/doc/65219}.

\bibitem{Segal}
I.~E. Segal.
\newblock ``Irreducible representations of operator algebras''.
\newblock \href{https://dx.doi.org/https://doi.org/10.1090/S0002-9904-1947-08742-5}{Bulletin of the American Mathematical Society {\bf 53}, 73--88}~(1947).

\bibitem{doran1994c}
R.S. Doran and A.M. Society.
\newblock ``C*-algebras: 1943-1993 : a fifty year celebration : Ams special session commemorating the first fifty years of c*-algebra theory, january 13-14, 1993, san antonio, texas''.
\newblock Contemporary mathematics. American Mathematical Society. ~(1994).

\bibitem{Aptekarev_1984}
A~I Aptekarev and E~M Nikishin.
\newblock ``The scattering problem for a discrete sturm-liouville operator''.
\newblock \href{https://dx.doi.org/https://doi.org/10.1070/SM1984v049n02ABEH002713}{Mathematics of the USSR-Sbornik {\bf 49}, 325}~(1984).

\bibitem{Duran_1996}
Antonio~J. Duran.
\newblock ``Markov’s theorem for orthogonal matrix polynomials''.
\newblock \href{https://dx.doi.org/https://doi.org/10.4153/CJM-1996-062-4}{Canadian Journal of Mathematics {\bf 48}, 1180–1195}~(1996).

\bibitem{Merkes1966}
E.~P. Merkes.
\newblock ``On truncation errors for continued fraction computations''.
\newblock \href{https://dx.doi.org/https://doi.org/10.1137/0703042}{SIAM Journal on Numerical Analysis {\bf 3}, 486--496}~(1966).

\bibitem{Koornwinder2013}
Tom~H. Koornwinder.
\newblock ``Orthogonal polynomials''.
\newblock \href{https://dx.doi.org/https://doi.org/10.1007/978-3-7091-1616-6_6}{Pages 145--170}.
\newblock Springer Vienna. ~(2013).

\bibitem{chihara2014introduction}
T.~S. Chihara.
\newblock ``An introduction to orthogonal polynomials''.
\newblock \href{https://dx.doi.org/https://doi.org/10.2307/3617920}{Gordon and Breach Science Publishers}. ~(1978).

\bibitem{Kempe2006}
Julia Kempe, Alexei Kitaev, and Oded Regev.
\newblock ``{The complexity of the local hamiltonian problem}''.
\newblock \href{https://dx.doi.org/https://doi.org/10.1137/S0097539704445226}{SIAM Journal on Computing {\bf 35}, 1070--1097}~(2006).

\bibitem{Jones1974}
William~B. Jones and W.~J. Thron.
\newblock ``Numerical stability in evaluating continued fractions''.
\newblock \href{https://dx.doi.org/https://doi.org/10.2307/2005701}{Mathematics of Computation {\bf 28}, 795--810}~(1974).

\bibitem{stanisic2022observing}
Stasja Stanisic, Jan~Lukas Bosse, Filippo~Maria Gambetta, Raul~A Santos, Wojciech Mruczkiewicz, Thomas~E O’Brien, Eric Ostby, and Ashley Montanaro.
\newblock ``Observing ground-state properties of the fermi-hubbard model using a scalable algorithm on a quantum computer''.
\newblock \href{https://dx.doi.org/https://doi.org/10.1038/s41467-022-33335-4}{Nature communications {\bf 13}, 5743}~(2022).

\bibitem{Epperly2022}
Ethan~N. Epperly, Lin Lin, and Yuji Nakatsukasa.
\newblock ``A theory of quantum subspace diagonalization''.
\newblock \href{https://dx.doi.org/https://doi.org/10.1137/21M145954X}{SIAM Journal on Matrix Analysis and Applications {\bf 43}, 1263--1290}~(2022).

\bibitem{lee2023sampling}
Gwonhak Lee, Dongkeun Lee, and Joonsuk Huh.
\newblock ``Sampling {E}rror {A}nalysis in {Q}uantum {K}rylov {S}ubspace {D}iagonalization''.
\newblock \href{https://dx.doi.org/10.22331/q-2024-09-19-1477}{{Quantum} {\bf 8}, 1477}~(2024).

\bibitem{Parker2019}
Daniel~E. Parker, Xiangyu Cao, Alexander Avdoshkin, Thomas Scaffidi, and Ehud Altman.
\newblock ``A universal operator growth hypothesis''.
\newblock \href{https://dx.doi.org/https://doi.org/10.1103/PhysRevX.9.041017}{Phys. Rev. X {\bf 9}, 041017}~(2019).

\bibitem{Kirby24}
William Kirby.
\newblock ``Analysis of quantum {K}rylov algorithms with errors''.
\newblock \href{https://dx.doi.org/https://doi.org/10.22331/q-2024-08-29-1457}{{Quantum} {\bf 8}, 1457}~(2024).

\bibitem{onorati2023efficient}
Cambyse Rouzé, Daniel Stilck~França, Emilio Onorati, and James~D. Watson.
\newblock ``Efficient learning of ground and thermal states within phases of matter''.
\newblock \href{https://dx.doi.org/https://doi.org/10.1038/s41467-024-51439-x}{Nature Communications{\bf 15}}~(2024).

\bibitem{onorati2023provably}
Emilio Onorati, Cambyse Rouzé, Daniel~Stilck França, and James~D. Watson.
\newblock ``Provably efficient learning of phases of matter via dissipative evolutions''~(2023).
\newblock  \href{http://arxiv.org/abs/2311.07506}{arXiv:2311.07506}.

\bibitem{De_Palma_2021}
Giacomo De~Palma, Milad Marvian, Dario Trevisan, and Seth Lloyd.
\newblock ``The quantum wasserstein distance of order 1''.
\newblock \href{https://dx.doi.org/https://doi.org/10.1109/tit.2021.3076442}{IEEE Transactions on Information Theory {\bf 67}, 6627–6643}~(2021).

\bibitem{chen2023local}
Chi-Fang Chen, Hsin-Yuan Huang, John Preskill, and Leo Zhou.
\newblock ``Local minima in quantum systems''~(2023).
\newblock  \href{http://arxiv.org/abs/2309.16596}{arXiv:2309.16596}.

\bibitem{Somma_2003}
Rolando~D. Somma, Gerardo Ortiz, Emanuel~H. Knill, and James Gubernatis.
\newblock ``Quantum simulations of physics problems''.
\newblock \href{https://dx.doi.org/https://doi.org/10.1117/12.487249}{SPIE}. ~(2003).

\end{thebibliography}

\end{document}